\documentclass[journal]{IEEEtran}
\usepackage[utf8]{inputenc}
\usepackage{amsmath}
\usepackage{amsthm}
\usepackage{array}
\usepackage{fixltx2e}
\usepackage{stfloats}
\usepackage{url}
\usepackage[table]{xcolor}
\usepackage{cite}
\usepackage{amssymb,amsfonts}
\usepackage{amsbsy}
\usepackage{algorithm}
\usepackage{graphicx}
\usepackage{graphics}
\usepackage{textcomp}
\usepackage{xcolor}
\usepackage{mathtools,amssymb,lipsum,nccmath,bm,amsmath}
\usepackage{cuted}
\usepackage{array}
\usepackage{ragged2e}
\usepackage[short]{optidef}
\usepackage{tabularx}
\usepackage{algpseudocode}
\usepackage{float}
\usepackage{tikz}
\usepackage{subcaption}
\usepackage{bbm}
\usepackage{multicol}
\newcommand{\mb}{\mathbf}

\newcommand\scalemath[2]{\scalebox{#1}{\mbox{\ensuremath{\displaystyle #2}}}}
\usepackage{nomencl}
\makenomenclature

\usepackage{empheq}
\usepackage[many]{tcolorbox}

\tcbset{
  highlight math style={
    colback=yellow,
    arc=0pt,
    outer arc=0pt,
    boxrule=0pt,
    top=2pt,
    bottom=2pt,
    left=2pt,
    right=2pt,
  }
}

\newcommand{\norm}[1]{\left\lVert #1 \right\rVert}

\newtheorem{theorem}{Theorem}
\newtheorem{remark}{Remark}

\newtheorem{lemma}[theorem]{Lemma}
\newtheorem{definition}{Definition}

\usepackage{tikz}
\usetikzlibrary{shapes, arrows.meta, positioning}

\usepackage{soul}

\IEEEoverridecommandlockouts
\allowdisplaybreaks[2]

\begin{document}
%

\title{A Novel Observer-Centric Approach for Detecting Faults in Islanded AC Microgrids with Uncertainties}

\author{
    \IEEEauthorblockN{Gabriel Intriago,~\IEEEmembership{Student Member,~IEEE}, Andres Intriago,~\IEEEmembership{Student Member,~IEEE},\\ Charalambos Konstantinou,~\IEEEmembership{Senior Member,~IEEE} and Yu Zhang,~\IEEEmembership{Member,~IEEE}}


    \thanks{G. Intriago is with the Department of Electrical and Computer Engineering at the University of California Santa Cruz, USA (Email: gintriag@ucsc.edu). He is also affiliated with the Department of Mathematics at the Escuela Politécnica Nacional, Ecuador (Email: gabriel.intriago01@epn.edu.ec).}
    \thanks{A. Intriago and C. Konstantinou are with the Computer, Electrical, and Mathematical Sciences and Engineering Division, King Abdullah University of Science and Technology, Saudi Arabia (Email: \{andres.intriagovelasquez,charalambos.konstantinou\}@kaust.edu.sa).}
    \thanks{Y. Zhang is with the Department of Electrical and Computer Engineering at the University of California, Santa Cruz, USA (Email: zhangy@ucsc.edu).}
}

\maketitle

\begin{abstract}
Fault detection is vital in ensuring AC microgrids' reliable and resilient operation. Its importance lies in swiftly identifying and isolating faults, preventing cascading failures, and enabling rapid power restoration. This paper proposes a strategy based on observers and residuals for detecting internal faults in grid-forming inverters with power-sharing coordination. The dynamics of the inverters are captured through a nonlinear state space model. The design of our observers and residuals considers $H_{-}/H_{\infty}$ conditions to ensure robustness against disturbances and responsiveness to faults. The proposed design is less restrictive than existing observer-based fault detection schemes by leveraging the properties of quadratic inner-boundedness and one-sided Lipschitz conditions. The internal faults considered in this paper include actuator faults, busbar faults, and inverter bridge faults, which are modeled using vector-matrix representations that modify the state space model of the inverters. One significant advantage of the proposed approach is its cost-effectiveness, as it does not require additional sensors. Experiments are conducted on an islanded AC microgrid with three inductive lines, four inductive loads, and four grid-forming inverters to validate the merits of the proposed fault detection strategy. The results demonstrate that our design outperforms existing methods in the field.


\end{abstract}


\nomenclature{\(m_{Pi}\)}{Droop coefficient for active power }
\nomenclature{\(n_{Qi}\)}{Droop coefficient for reactive power}
\nomenclature{\(L_{ci}\)}{Inductance of the output connector}
\nomenclature{\(R_{ci}\)}{Resistance of the output connector}
\nomenclature{\(C_{fi}\)}{Capacitance of the LC filter}
\nomenclature{\(L_{fi}\)}{Inductance of the LC filter}
\nomenclature{\(R_{fi}\)}{Resistance of the LC filter}
\nomenclature{\(K_{IV_i}\)}{Integral gain of the voltage PI controller}
\nomenclature{\(K_{IC_i}\)}{Integral gain of the current PI controller}
\nomenclature{\(K_{PV_i}\)}{Proportional gain of the voltage PI controller}
\nomenclature{\(K_{PC_i}\)}{Proportional gain of the current PI controller}
\nomenclature{\(v_{idi,iqi}^*\)}{Output dq voltage of the PI current controller}
\nomenclature{\(i_{ldi,lqi}^*\)}{Output dq current of the PI voltage controller}
\nomenclature{\(v_{odi,oqi}^*\)}{Reference dq voltages for the PI voltage controller}
\nomenclature{\(v_{bdi,bqi}\)}{dq voltage of the output connector }
\nomenclature{\(V_{ni}\)}{Voltage magnitude setpoint}
\nomenclature{\(\omega_{ni}\)}{Frequency setpoint}
\nomenclature{\(\omega_{ci}\)}{Low-pass filters cut-off frequency for the active and reactive power}
\nomenclature{\(\omega_{b}\)}{Nominal frequency of the system}
\nomenclature{\(\omega_{com}\)}{Frequency of the common reference frame}
\nomenclature{\(\omega_{i}\)}{Operating frequency of each inverter}
\nomenclature{\(i_{odi,oqi}\)}{dq output current of the LC filter}
\nomenclature{\(v_{odi,oqi}\)}{dq output voltage of the LC filter}
\nomenclature{\(i_{ldi,lqi}\)}{dq output current of the inverter bridge }
\nomenclature{\(\gamma_{di,qi}\)}{dq auxiliary state variables for the current PI controllers}
\nomenclature{\(\phi_{di,qi}\)}{dq auxiliary state variables for the voltage PI controllers}
\nomenclature{\(Q_{i}\)}{Reactive power of each inverter}
\nomenclature{\(P_{i}\)}{Active power of each inverter}
\nomenclature{\(\alpha_{i}\)}{Rotating power angle of each inverter}
\nomenclature{\(t_{o}\)}{Fault detection response time}
\nomenclature{\(t_{c}\)}{Fault clearing response time}

\printnomenclature

\section{Introduction}
The introduction of converter-interfaced power generation has successfully integrated renewable energy sources into small-scale power systems, such as microgrids \cite{Fahimi2011} \cite{Ku2019}. Microgrids can operate in two modes: grid-connected or islanded. In the grid-connected mode, a microgrid relies on the main grid for voltage and frequency regulation, which benefits from its upstream protection. On the other hand, an islanded microgrid functions as an independent system and must independently maintain reference voltage magnitude and frequency for its components. Consequently, ensuring stability in the islanded mode is more challenging than in the grid-connected mode \cite{Farrokhabadi2020}. In addition, independent microgrids face risks to their stability during abnormal events, such as internal faults, which can result in significant imbalances between energy demand and supply. These imbalances can lead to partial power outages or blackouts. Furthermore, when an islanded microgrid disconnects from the main power system, the fault current strength decreases, which helps mitigate the decline in voltage magnitude and frequency following a severe event \cite{Ramasubramanian2019}. This situation is exacerbated when the power generators are grid-forming inverters, as they reduce the available electrical inertia in the system \cite{Lasseter2020}. To ensure stable microgrid performance and prevent interruptions in energy supply, it is crucial to have smart fault detection (FD) that is resilient against disturbances and sensitive to faults.


\subsection{Prior work}
Various fault detection schemes are implemented in microgrids to ensure reliable and efficient operation of AC microgrids \cite{Bansal2018}. These schemes have gained popularity, particularly those employing data-driven techniques, owing to advancements in storage, software, handling, and hardware devices \cite{Srivastava2022,Baghaee2021,Baghaee2020}.  However, the success of such methods heavily relies on the quantity and quality of the data they utilize.

Signal processing-based fault detection schemes have proven effective in identifying fault signatures using non-parametric techniques. However, these techniques are often limited to specific conditions and scenarios \cite{Jarrahi2020,Yadav2019,Sadeghkhani2018}. Another approach involves utilizing signal patterns and local measurements near the generation units to create a dependable fault detection module \cite{Sharma2022,Bhargav2022,Xiao2022}. Nevertheless, these strategies may require additional hardware and incur additional costs. In addition to these techniques, model-based methods are employed for fault detection when a system model is available. These methods have gained popularity due to their reliance on the physical relationships governing system dynamics. Among the model-based approaches, observer-based methods have garnered significant interest due to their fast detection capabilities, cost-effectiveness, and the availability of powerful tools for observer design \cite{Campos2011,Hawkins2020,Wang2021,Anagnostou2018,Shoaib2022}.

In \cite{Hawkins2020}, the authors propose a nonlinear observer for a grid-connected photovoltaic (PV) circuit that monitors unmodeled fault signatures to detect output deviations that could indicate the presence of a fault. However, their study does not account for disturbances and uncertainties in the system parameters. \cite{Wang2021} presents an improved fault detection and identification method using $\mathcal{H}_{-}/\mathcal{H}_{\infty}$ optimization, which can effectively handle various component faults. However, their study focuses on DC microgrids with a linear state-space model.
Anagnostou et al.~\cite{Anagnostou2018} develop a time-varying observer for predicting the states of synchronous machines. Their approach linearizes the model at each time step, leading to increased computational complexity. Similarly, they do not consider model uncertainties and disturbances.
\cite{Shoaib2022} introduces a constrained minimization program with linear matrix inequality (LMI) constraints and $\mathcal{L}_{-}/\mathcal{L}_{\infty}$ performance indices to detect faults in a microgrid comprising synchronous machines. They utilize a Lipschitz equivalent nonlinear model. Their study neglects parametric uncertainties and the influence of faults on the nonlinear function in the state transition model. Moreover, the restrictiveness of the Lipschitz condition is not addressed regarding the feasibility of the LMI constraints.

The primary limitation of fault detection strategies based on the Lipschitz condition for nonlinear systems is their susceptibility to the Lipschitz constant, resulting in conservative observer designs \cite{Abbaszadeh2010,Zhang2012}. In \cite{Abbaszadeh2010}, an observer design problem is introduced for nonlinear systems that take into account the quadratic inner-boundedness (QB) and one-sided Lipschitz (OL) conditions. \cite{Zhang2012} presents reduced- and full-order observer designs for nonlinear systems satisfying the QB and OL conditions, incorporating the Ricatti equation. However, these previous works do not specifically focus on fault detection and do not consider the robustness against disturbances and the responsiveness to faults.

\subsection{Contributions}
Considering the limitations of prior works, the main contributions of our work are summarized as follows. First, our research demonstrates the successful application of LMI-based techniques combined with $\mathcal{H}_{-}/\mathcal{H}_{\infty}$ optimization for fault detection in droop-controlled grid-forming converters (GFM) connected to AC microgrids in islanded mode. The proposed fault detection strategy, based on observers, exhibits robustness against disturbances and responsiveness to faults.

To the best of our knowledge, this is the first approach that proposes a GFM nonlinear model incorporating the one-sided Lipschitz (OL) and quadratic inner-boundedness (QB) conditions. We thoroughly investigate the relationship between the observer designs based on Lipschitz conditions and the proposed design based on the OL and QB conditions. Furthermore, we derive a deterministic matrix representation for various types of faults affecting GFMs. Our analysis considers the influence of parametric uncertainties and disturbances on the system model, including the nonlinear function in the state transition equation.

Utilizing a nonlinear observer offers the advantage of performing the one-time computation of the observer's gain offline. In contrast, a linearized observer design increases the time complexity by recalculating the observer gain at each time step \cite{Anagnostou2018}. Unlike the Lipschitz-based observer design in \cite{Shoaib2022}, the OL and QB conditions allow for a less restrictive observer design, avoiding the sensitivity of the observer to the Lipschitz constant. One major challenge in grid-forming inverter technology is the influence of different controller designs. Our proposed method also possesses a crucial advantage in that it can be implemented with the two most commonly used practical control techniques: droop control and virtual synchronous machine control. While this work primarily focuses on detecting internal faults in grid-forming inverters, we believe that our observer design can be easily adapted to detect internal faults in grid-following inverters, given the similar dynamics shared by both technologies. These advantages greatly enhance our proposed method's value for real-world applications.

\subsection{Notation and organization}
$\mathbb{R}$ denotes the set of real numbers. Symbol $\times$ denotes the Cartesian product. Bold letters represent vectors. $\norm{\mb{a}}$ is the $\ell_2$-norm of the vector while $\|\mb{a}\|_{[0,t]} := \sqrt{\int_0^t\|\mb{a}\|_2^2\:dt}$ is the $\ell_2$-norm up to time $t$. $A \succ 0$ indicates that matrix $A$ is positive definite. $\langle\mb{a}, \mb{b}\rangle$ represents the dot product between vectors $\mb{a}$ and $\mb{b}$. $\mb{a}^{\top}$ is the transpose of  vector $\mb{a}$.

The rest of the paper is organized as follows. Section \ref{preliminaries} introduces the background for observer design in nonlinear systems and residual-based fault detection. Section \ref{inverter-model} presents the models of grid-forming inverters and different faults. Section \ref{proposed} outlines the Lipschitz and proposed observer design. The numerical tests and discussion are given in Sections \ref{numerical-tests} and \ref{discussion}, respectively. The concluding remarks are presented in Section \ref{conclusion}.

\section{Preliminaries}\label{preliminaries}
\subsection{Observers for nonlinear dynamic systems}
The state-space representation of a nonlinear system can be typically expressed as:
\begin{align}
    \left\{
    \begin{matrix}
        \dot{\mb{x}} = A\mb{x} + B\mb{u} + \boldsymbol{\phi}(\mb{u},\mb{x}) \\
        \mb{y} = C\mb{x} + D\mb{u} \qquad \qquad \qquad
    \end{matrix}
    \right.\, , \label{basic-state-space-model}
\end{align}
where $\bm{y}$ represents the measurement vector, $\bm{x}$ is the state vector, $\bm{u}$ the inputs vector, and function $\boldsymbol{\phi}(\mb{x},\mb{u})$ captures nonlinearity of the system. The parametric matrices $A$, $B$, $C$, and $D$ are formed according to the equations that govern the system.  We further define two convex sets $\mathcal{D}$ and $\mathcal{U}$ that make up the feasible operating region of the system. They are the Cartesian products of all intervals formed by the maximum and minimum limits of the entire control inputs and states, respectively \cite{Qi2018}:
\begin{align*}
    \mathcal{D} &= [x_1^{\min},x_1^{\max}]\times\dots\times[x_n^{\min},x_n^{\max}], \\
    \mathcal{U} &= [u_1^{\min},u_1^{\max}]\times\dots\times[u_p^{\min},u_p^{\max}].
\end{align*}
The nonlinear model~\eqref{basic-state-space-model} can be modified to incorporate both disturbances and faults as follows:
\begin{equation}
    \begin{cases}
        \dot{\mb{x}} &= A\mb{x} + B\mb{u} + \boldsymbol{\phi}(\mb{x},\mb{u}) + E_w\mb{w} + E_f\mb{f} + \eta_1(\mb{x},\mb{u}) \\
        \mb{y} &= C\mb{x} + D\mb{u} + F_w\mb{w} + F_f\mb{f} + \eta_2(\mb{x},\mb{u}),
    \end{cases}\, \label{system-model-with-disturb-faults}
\end{equation}
where $\mb{w}$ is a disturbance input vector, $\mb{f}$ is the fault vector, $E_w$ and $F_w$ are constant disturbance matrices, $E_f$ and $F_f$ are deterministic fault matrices. The unknown functions $\eta_1$ and $\eta_2$ represent the model uncertainties which are considered to be bounded and finite, i.e.,
\[\|\eta_i(\mb{u},\mb{x})\| < \infty,\quad \forall\:\mb{x} \in \mathcal{D}~\text{and}~\forall\:\mb{u}\:\in\mathcal{U},~i =1,2.\]
A Luenberger observer of the nonlinear system is defined as follows:
\begin{equation}
    \begin{cases}
        \dot{\hat{\mb{x}}} &= A\mb{\hat{x}} + B\mb{u} + \phi(\mb{\hat{x}},\mb{u}) + L(\mb{y} - \hat{\mb{y}}) \\
        \hat{\mb{y}} &=  C\mb{\hat{x}} + D\mb{{u}}
    \end{cases}\,,
    \label{luenberger-observer}
\end{equation}
where $\hat{\mb{y}}$ is the vector of estimated measurements, $\mb{\hat{x}}$ is the vector of estimated states, and the matrix $L$ is the observer gain to be designed. We have two assumptions of model \eqref{system-model-with-disturb-faults}:
\begin{enumerate}
\item The system is observable.
\item The signal vectors $\mb{w}$ and $\mb{f}$are $\mathcal{L}_2$ square-integrable that satisfy:
$\|\mb{w}\|_{[0,t]} <\infty$ and $
\|\mb{f}\|_{[0,t]} <\infty$.
\end{enumerate}

\subsection{Residual generation}
Define the error of state $\mb{e} = {{\mb{x}}} - {\hat{\mb{x}}}$ and measurement residual $\mb{r} = W(\mb{y} -\hat{\mb{y}})$. Then, the dynamics of the error of state can be derived as follows:
\begin{equation}
    \begin{cases}
        \dot{\mb{e}} &= (A - LC)\mb{e} + (E_w - LF_w)\mb{w}  \\
        &\quad + (E_f - LF_f)\mb{f} + \phi(\mb{x},\mb{u}) - \phi(\hat{\mb{x}},\mb{u})\\
        \mb{r} &=W(\mb{y} - C\mb{\hat{x}} - D\mb{{u}})\quad 
    \end{cases}. \label{error-dynamics}
\end{equation} 
We can compactly rewrite \eqref{error-dynamics} as:
\begin{equation}
    \begin{cases}
        \dot{\mb{e}} &= \bar{A}\mb{e} + \Phi + \bar{E}_w\mb{w} + \bar{E}_f\mb{f} \\
        \mb{r} &= \bar{C}\mb{e} + \bar{F}_w\mb{w} + \bar{F}_f\mb{f}
        \label{error-dynamics-2}
    \end{cases}
\end{equation}
with $\Phi \triangleq \phi(\mb{x},\mb{u}) - \phi(\hat{\mb{x}},\mb{u})$ and
\begin{align}
    \left\{\begin{matrix}
        \bar{A} = A - LC, & \bar{E}_w = E_w - LF_w, & \bar{E}_f = E_f - LF_f\\
        \bar{C} = WC, & \bar{F}_w = WF_w, & \bar{F}_f = WF_f
    \end{matrix}\right. .
    \label{definitions}
\end{align}
We assume $W=I$ as suggested by \cite{Buciakowski2017}. The residual $\mb{r}$ can be used for fault detection \cite{Shoaib2022}.

\subsection{Mixed $\mathcal{H}_{-}/\mathcal{H}_{\infty}$ optimization for observer design}
The fault detection problem becomes challenging when there is no clear distinction between disturbances and faults. Such a situation may mislead the fault detection filter, triggering false alarms. To cope with such a concern, we considered the mixed $\mathcal{H}_{-}/\mathcal{H}_{\infty}$ optimization framework for designing our proposed observer \cite{Shoaib2022}. This framework aims to simultaneously make the fault detector filter responsive to faults and sturdy against disturbances by satisfying the following criteria
\begin{enumerate}
    \item Robustness against disturbances $\mb{w}$:
        \begin{align}\|\mb{r}_{w}\|_{[0,t]} \leq \alpha\|\mb{w}\|_{[0,t]}\label{rob_cond}.\end{align}
    \item Sensitivity to faults $\mb{f}$:
        \begin{align}\|\mb{r}_{f}\|_{[0,t]} \geq \beta\|\mb{f}\|_{[0,t]}\label{sens_cond}\end{align}
\end{enumerate}
where
\begin{align}
    \mb{r}_{w} &= \bar{C}\mb{e} + \bar{F}_w\mb{w}\label{fault_free_res}\\
    \mb{r}_{f} &= \bar{C}\mb{e} + \bar{F}_f\mb{f}. \label{dist_free_res}
\end{align}

In Section \ref{proposed}, we pose the task of designing a fault detection filter as a convex optimization problem. We will find linear matrix inequalities to propose sufficient conditions to guarantee the existence of such a filter. In this sense, we transform conditions \eqref{rob_cond} and \eqref{sens_cond} into linear matrix inequalities to incorporate them in the filter design effectively.

\subsection{Triggering method of the proposed scheme}
We follow the procedure presented in \mbox{\cite{Ding2013}} to decide the triggering method that detects the occurrence of faults. The triggering mechanism consists of three important components: the residual norm $J=\|\mb{r}\|_2$, the threshold $J_{th}$, and a signal comparator that acts based on the magnitudes of $J$ and $J_{th}$. The first component, the comparator, triggers an alarm when the residual norm exceeds the threshold, as indicated by the following detection logic
\begin{equation}
    \begin{cases}
    J \leq J_{\mathrm{th}} &\longrightarrow\:\text{No faults} \\
    J > J_{\mathrm{th}} &\longrightarrow\:\text{Fault detected}
    \end{cases}.\label{fdi-logic}
\end{equation}

The second component, the threshold, is necessary because disturbances make $J$ non-zero even without faults. The threshold is calculated when the system is in a fault-free condition and subjected to disturbances
\begin{align}
    J_{\mathrm{th}} = \underset{\mb{w}\in\mathcal{L}_2}{\sup}\: \|\mb{r}_{\mb{w}}\|. \label{residual-threshold}
\end{align} 
The threshold is an upper bound for the residual norm in a fault-free condition. Notice that the threshold is a fixed magnitude immutable to the faults and the GFM's conditions. The third component, the residual norm, is computed from the proposed observers, which are designed to significantly perturb the residual norms by the presence of faults, making the residual norm exceed the threshold.

\section{Models of inverters and faults}\label{inverter-model}

\subsection{Dynamic model of grid-forming inverters}
\begin{figure}
    \centering
    \includegraphics[width=0.8\linewidth]{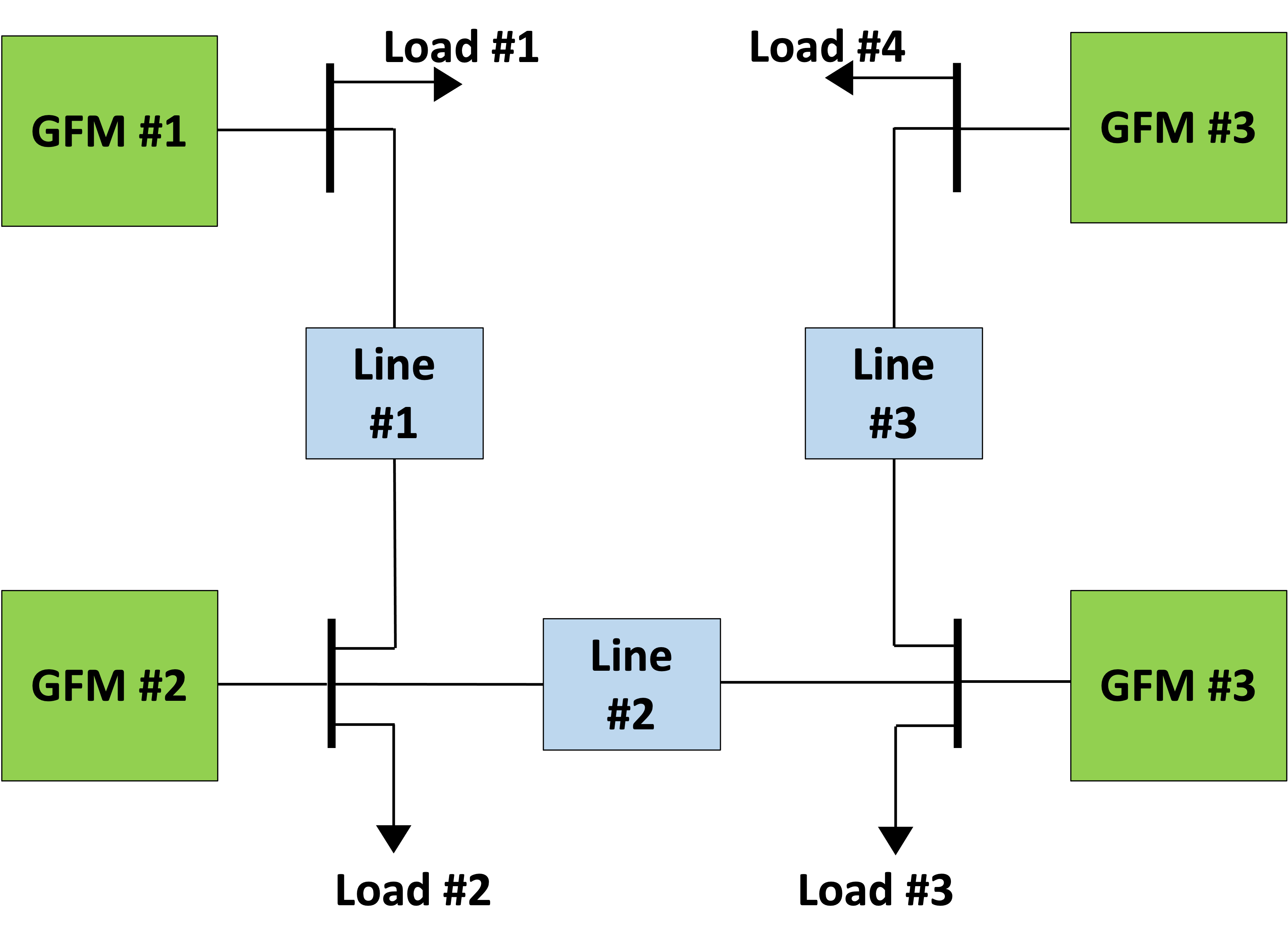}
    \caption{An islanded droop-controlled AC microgrid test system \cite{Bidram2013}.}
    \label{test-system}
    \vspace{-0.2cm}
\end{figure}

We use the $dq$ reference frame to represent the dynamics of the GFMs according to \cite{Bidram2013} where the $i$-th GFM is indexed with the subscript $i$. The dynamical model of each GFM in state space form considering disturbance- and fault-free scenarios is
\begin{align}
    \left\{
    \begin{matrix}
        \dot{\mb{x}} = A\mb{x} + B\mb{u} + \boldsymbol{\phi}(\mb{x},\mb{u}) \\ 
        \mb{y} = C\mb{x} + D\mb{u} \qquad \qquad \qquad
    \end{matrix}
    \right.\,, \label{system-model}
\end{align}
where
\begin{align*}
\mb{x} &= \left[\alpha_i \:\: P_i \:\: Q_i \:\: \phi_{di} \:\: \phi_{qi} \:\: \gamma_{di} \:\: \gamma_{qi} \:\:  i_{ldi} \:\: i_{lqi} \:\:  v_{odi} \:\: v_{oqi} \:\: i_{odi} \:\: i_{oqi} \right]^\top \\
\mb{u} &= \left[\omega_{com} \:\: \omega_{ni} \:\: V_{ni} \:\: v_{bdi} \:\: v_{bqi}\right]^\top \\
\mb{y} &= \left[\alpha_i \:\: \omega_i \:\: v_{odi}^* \:\: i_{ldi}^* \:\: i_{lqi}^* \:\: v_{idi}^* \:\: v_{iqi}^* \right]^\top.
\end{align*}
The nonlinear function $\boldsymbol{\phi}(\mb{x},\mb{u})$ and the parametric matrices $A$, $B$, $C$, $D$ are formed according to the differential equations governing the GFMs as shown in \cite{Bidram2013,Tito2023,Intriago2023}. The function $\boldsymbol{\phi}(\mb{x},\mb{u})$ in \eqref{system-model} arises from the Kirchoff voltage and current laws and the low-pass filtered reactive and active power.

Next, we introduce four fault models defined by the fault vector $\mb{f}$ and the fault matrices $E_f$ and $F_f$. The fault models are not derived from the theory of small signal analysis. Instead, they are modeled according to the theory of model-based fault diagnosis \mbox{\cite{Ding2013}}. Such modeling maps faults as alterations in the inverter's state space model, allowing the construction of the fault vector $\mb{f}$ and fault matrices $E_f$, $F_f$ \mbox{\cite{Shoaib2022,Wang2021}}. Mathematically, the alterations in the GFM state space representation originate from the offset sudden change in the variables of interest modeled as $z^\prime = z + \Delta z$ where $z$ can be $v_{bdi}$ and $v_{bqi}$ for busbar faults, $\omega_{ni}$ and $V_{ni}$ for actuator faults, or $\eta_{v_{idi}}$ and $\eta_{v_{iqi}}$ for bridge faults.

\subsection{Busbar faults}
A busbar fault affects the voltage phasor at the point of common coupling (PCC) with the microgrid. Such voltage is represented by the inputs $v_{bdi}$ and $v_{bqi}$. The busbar fault is implemented as a symmetrical fault corresponding to the grounding of the point of common coupling (busbar that connects the GFM with the microgrid, also known as PCC) through a balanced low-resistance branch (0.1 $\Omega$). The sudden change in the bus voltage at the PCC is modeled as follows
\begin{align}
    v_{bdi}^\prime &= v_{bdi} + \Delta v_{bdi} \label{3phasef-vbdi}\\
    v_{bqi}^\prime &= v_{bqi} + \Delta v_{bqi}\label{3phasef-vbqi}
\end{align}
Hence, by plugging {\eqref{3phasef-vbdi}} and {\eqref{3phasef-vbqi}} into {\eqref{system-model}}, the fault model for the busbar fault is given as
\begin{align}
    f &= \left[
                \begin{matrix}
                        \Delta v_{bdi} & \Delta v_{bqi} 
                \end{matrix}
            \right]^\top, \notag\\
    E_f &= 
    \left[
        \begin{matrix}
            0_{1\times11} & -\frac{1}{L_{ci}} & 0 \\
            0_{1\times11} & 0 & -\frac{1}{L_{ci}}
        \end{matrix}
    \right]^\top, \notag\\
    F_f &= 
    \left[
        \begin{matrix}
            0_{7\times 2}
        \end{matrix}
    \right]. \notag
\end{align}

\subsection{Actuator faults}
The actuator signals consist of the inputs $\omega_{ni}$ and $V_{ni}$, respectively, which set the GFMs' desired frequency and voltage magnitude. A sudden 10\% change in the actuator reference signals is modeled as follows
\begin{align}
    \omega_{ni}^\prime &= \omega_{ni} + \Delta \omega_{ni} \label{actuatorf-omegani}\\
    V_{ni}^\prime &= V_{ni} + \Delta V_{ni}.\label{actuatorf-Vni}
\end{align}
The actuator fault on $\omega_{ni}$ affects the linear term $B\mb{u}$ and the nonlinear function $\boldsymbol{\phi}(\mb{x},\mb{u})$ of \eqref{system-model}. Hence, by replacing {\eqref{actuatorf-omegani}} in {\eqref{system-model}}, the fault model of the actuator $\omega_{ni}$ fault is
\begin{align}
    f_{\omega_{ni}} &= \left[
            \begin{matrix}
                    \Delta\omega_{ni} & \Delta\omega_{ni}i_{lqi} & \Delta\omega_{ni}i_{ldi} & \Delta\omega_{ni}v_{oqi}
            \end{matrix}
        \right. \notag \\
        &\qquad\left. 
            \begin{matrix}
                \Delta\omega_{ni}v_{odi} & \Delta\omega_{ni}i_{oqi} & \Delta\omega_{ni}i_{odi}
            \end{matrix}
        \right]^\top
        \notag\\
    E_{f_{\omega_{ni}}} &=
    \left[
        \begin{matrix}
            1 & 0_{1\times6} \\
            0_{6\times1} & 0_{6\times6} \\
            0_{6\times1} & A_{\omega_{ni}}
        \end{matrix}
    \right] \notag\\
    A_{\omega_{ni}} &= \text{diag}\left([\begin{matrix}
        1 & -1 & 1 & -1 & 1 & -1
    \end{matrix} ]\right) \notag\\
    F_{f_{\omega_{ni}}} &=
    \left[
        \begin{matrix}
            0 & 1 & 0_{1\times 5} \\
            0_{6\times 1} & 0_{6\times 1} & 0_{6\times 5}
        \end{matrix}
    \right]^\top \notag
\end{align}
The actuator fault on $V_{ni}$ affects the linear term $B\mb{u}$ of \eqref{system-model} exclusively. Similarly, by replacing {\eqref{actuatorf-Vni}} in {\eqref{system-model}}, the fault model of the actuator $V_{ni}$ fault is
\begin{align}
    f_{V_{ni}} &= \left[
            \begin{matrix}
                    \Delta V_{ni}
            \end{matrix}
        \right] \notag\\
    E_{f_{V_{ni}}} &=
    \left[
        \begin{matrix}
            0_{1\times 3} & 1 & 0 & K_{PV_i} & 0 & \frac{1}{L_{fi}}K_{PC_i}K_{PV_i} & 0_{1\times 5}
        \end{matrix}
    \right]^\top \notag\\
    F_{f_{V_{ni}}} &=
    \left[
        \begin{matrix}
            0 & 0 & 1 & K_{PV_i} & 0 & K_{PC_i}K_{PV_i} & 0
        \end{matrix}
    \right]^\top \notag
\end{align}

In this study, the actuator $\omega_{ni}$ and $V_{ni}$ signals are the outputs of secondary controllers designed for the regulation of both reference signals. We consider that the actuator faults occur due to malfunctioning of the secondary controllers' input signals rather than the controllers themselves. The actuator signals these controllers compute rely on measurements such as output voltage magnitude $v_{oi}$ and frequency $\omega_i$ of each inverter \mbox{\cite{Xiao2023,BidramMulti2014}}. A physical fault or the malfunctioning of the devices that measure $v_{oi}$ and $\omega_i$ impacts the output of the secondary controller and, consequently, the grid-forming inverter's performance \mbox{\cite{Bidram2013}}.

\subsection{Inverter bridge faults}
The inverter bridge fault is modeled as an abrupt 10\% change in the inverter's efficiency. Ideally, the DQ output voltages of the current controller $v_{idi}^*$ and $v_{iqi}^*$ are equal to the inverter bridge's output voltage $v_{idi}$ and $v_{iqi}$
\begin{align*}
    v_{idi} &= \eta_{v_{idi}}v_{idi}^*,\quad \eta_{v_{idi}} = 1 \\
    v_{iqi} &= \eta_{v_{iqi}}v_{iqi}^*,\quad \eta_{v_{iqi}} = 1
\end{align*}
In this sense, we model the sudden change in the inverter bridge's efficiency as a parametric abrupt change
\begin{align*}
    v_{idi} &= (1 - \Delta\eta_{v_{idi}})v_{idi}^*,\quad \Delta\eta_{v_{idi}} \in (0,1]\\
    v_{iqi} &= (1 - \Delta\eta_{v_{iqi}})v_{iqi}^*,\quad \Delta\eta_{v_{iqi}} \in (0,1]
\end{align*}
We find the inverter bridge fault model by modifying the corresponding dynamical equations that depend on $v_{idi}^*$ and $v_{iqi}^*$. Consequently, we split the analysis into two parts.
\subsubsection{Fault vector and fault matrices for $v_{idi}$}
\begin{align*}
    \mb{f}_{v_{idi}} &= \Delta\eta_{v_{idi}}
        \left[
            \begin{matrix}
                    Q_i & \phi_{di} & \gamma_{di} & i_{ldi} & i_{lqi} & v_{odi}
            \end{matrix}
        \right. \\
        &\qquad \qquad \qquad
        \left.
            \begin{matrix}
                v_{oqi} & i_{odi} & V_{ni}
            \end{matrix}
        \right]^\top \\
    E_{f_{v_{idi}}} &= 
    \left[
        \begin{matrix}
            0_{9\times7} & \boldsymbol{\xi}_{v_{idi}} & 0_{9\times5}
        \end{matrix}
    \right]^\top\\
    F_{f_{v_{idi}}} &= 
    \left[
        \begin{matrix}
            0_{9\times5} & \boldsymbol{\tau}_{v_{idi}} & 0_{9\times1}
        \end{matrix}
    \right]^\top\\
    \boldsymbol{\xi}_{v_{idi}} &= 
        \left[
            \begin{matrix}
                \frac{K_{PC_i}K_{PV_i}n_{Q_i}}{L_{fi}} & - \frac{K_{PC_i}K_{IV_i}}{L_{fi}} & - \frac{K_{IC_i}}{L_{fi}} & \frac{K_{PC_i}}{L_{fi}} & \omega_b 
            \end{matrix}
        \right. \\
        & \quad 
        \left.
            \begin{matrix}
                \frac{K_{PC_i}K_{PV_i}}{L_{fi}} & \frac{K_{PC_i}\omega_bC_{fi}}{L_{fi}} & - \frac{K_{PC_i}F_i}{L_{fi}} & - \frac{K_{PC_i}K_{PV_i}}{L_{fi}}
            \end{matrix}
        \right]^\top \\
    \boldsymbol{\tau}_{v_{idi}} &= 
        \left[
            \scalemath{1}{
            \begin{matrix}
                K_{PC_i}K_{PV_i}n_{Q_i} & -K_{PC_i}K_{IV_i} & -K_{IC_i} & K_{PC_i}
            \end{matrix}
            }
        \right.\\
        & \quad \quad
        \left.
        \begin{matrix}
            \omega_b L_{fi} & K_{PC_i}K_{PV_i} & K_{PC_i}\omega_bC_{fi} & - K_{PC_i}F_i
        \end{matrix}
        \right. \\
        & \quad \quad
        \left.
        \begin{matrix}
            -K_{PC_i}K_{PV_i}
        \end{matrix}
        \right]^\top
\end{align*}
\subsubsection{Fault vector and fault matrices for $v_{iqi}$}
\begin{align*}
    \mb{f}_{v_{iqi}} &= \Delta\eta_{v_{iqi}}
        \left[
            \begin{matrix}
                    \phi_{qi} &
                    \gamma_{qi} &
                    i_{ldi} &
                    i_{lqi} &
                    v_{odi} &
                    v_{oqi} &
                    i_{oqi}
            \end{matrix}
        \right]^\top \\
    E_{f_{v_{iqi}}} &= 
    \left[
        \begin{matrix}
            0_{7\times8} & \boldsymbol{\xi} & 0_{7\times4}
        \end{matrix}
    \right]^\top \\
    F_{f_{v_{iqi}}} &= 
    \left[
        \begin{matrix}
            0_{7\times6} & \boldsymbol{\tau}_{v_{iqi}}
        \end{matrix}
    \right]^\top\\
    \boldsymbol{\xi}_{v_{iqi}} &= 
        \left[
            \begin{matrix}
                - \frac{K_{PC_i}K_{IV_i}}{L_{fi}} & - \frac{K_{IC_i}}{L_{fi}} & - \omega_b & \frac{K_{PC_i}}{L_{fi}}
            \end{matrix}
        \right. \\
        & \quad \quad
        \left.
            \begin{matrix}
                - \frac{K_{PC_i}\omega_b C_{fi}}{L_{fi}} & \frac{K_{PC_i}K_{PV_i}}{L_{fi}} & - \frac{K_{PC_i}F_i}{L_{fi}}
            \end{matrix}
        \right] \\
    \boldsymbol{\tau}_{v_{iqi}} &= 
        \left[
            \scalemath{1}{
            \begin{matrix}
                - K_{PC_i}K_{IV_i} & -K_{IC_i} & -\omega_bL_{fi} & K_{PC_i}
            \end{matrix}
            }
        \right. \\
        & \quad \quad
        \left.
        \begin{matrix}
            -K_{PC_i}\omega_bC_{fi} & K_{PC_i}K_{PV_i} & - K_{PC_i}F_i
        \end{matrix}
        \right]^{\top}.
\end{align*}
The inverter bridge fault affects both efficiencies simultaneously. The fault model of the inverter bridge fault is given as
$\mb{f} = 
    \left[
        \begin{matrix}
            \mb{f}_{v_{idi}}, \mb{f}_{v_{iqi}}
        \end{matrix}
    \right]^{\top}$, 
    $E_{f} = 
    \left[
        \begin{matrix}
            E_{f_{v_{idi}}}, E_{f_{v_{iqi}}}
        \end{matrix}
    \right]$, and 
    $F_{f} = 
    \left[
        \begin{matrix}
            F_{f_{v_{idi}}},  F_{f_{v_{iqi}}}
        \end{matrix}
    \right]$.


\section{Proposed observer design}\label{proposed}
\subsection{Lipschitz observer design}
Lipschitz systems include a broad range of physical systems providing flexibility for the design of observers. Let us assume the nonlinear function $\boldsymbol{\phi}(\cdot)$ satisfies the Lipschitz condition given by the definition given as follows:
\begin{definition}\label{lipschitz-definition}
A function $\boldsymbol{f}(\cdot)$ is \textit{Lipschitz} continuous if there exists a constant $\gamma >0$ such that $\forall\:\mb{u}\in\mathcal{U}$ and $\forall\:\mb{x},\hat{\mb{x}}\in\mathcal{D}$ we have
\begin{align}
    \|\boldsymbol{f}(\mb{x},\mb{u})-\boldsymbol{f}(\hat{\mb{x}},\mb{u})\| \leq \gamma \|\mb{x}-\hat{\mb{x}}\|. \label{lipschitz-condition}
\end{align}
\end{definition}
The lemma presented below gives sufficient conditions for the existence of a filter gain matrix $L$ based on the mixed $\mathcal{H}_{-}/\mathcal{H}_{\infty}$ optimization framework for Lipschitz nonlinear systems.

\begin{lemma}[see \cite{Shoaib2022}]\label{lipschitz-proposition}
Consider the system \eqref{system-model-with-disturb-faults} with a Lipschitz continuous observer defined in \eqref{luenberger-observer} and  \eqref{lipschitz-condition}. If exist a gain matrix $L$, strictly positive scalars $\alpha$, $\beta$, $\varepsilon_1$, $\varepsilon_2$, and identical positive definite matrices $Q$ and $P$ such that the following two LMIs hold:
\begin{align}
\begin{split}
    \left[\begin{matrix}
    \Omega_1 & P\bar{E}_w + \bar{C}^\top\bar{F}_w & P \\
    * & - \alpha^2I + \bar{F}_w^\top\bar{F}_w & 0 \\
    * & * & -\varepsilon_1I
    \end{matrix}\right] \prec 0 
\end{split}\label{L1-lipschitz}
\end{align}
\begin{align}
\begin{split}
    \left[\begin{matrix}
        \Omega_2 & Q\bar{E}_f - \bar{C}^\top\bar{F}_f & Q \\
        * & - \beta^2I + \bar{F}_f^\top\bar{F}_f & 0 \\
        * & * & -\varepsilon_2I
    \end{matrix}\right] \prec 0, 
\end{split}\label{L2-lipschitz}
\end{align}
where $\Omega_1 = \bar{A}^\top P + P\bar{A} + \bar{C}^\top\bar{C} + \varepsilon_1\gamma^2I$, $\Omega_2 = \bar{A}^\top Q + Q\bar{A} - \bar{C}^\top\bar{C} + \varepsilon_2\gamma^2I$, and 
$\gamma$ is the given Lipschitz constant. Then, i) the residual generator in \eqref{error-dynamics} is stable; and ii) the resulting observer satisfies the robustness and sensitivity constraints \eqref{rob_cond} and \eqref{sens_cond}.
\end{lemma}
The proof of \textit{Lemma} \ref{lipschitz-proposition} can be found in \cite{Shoaib2022}, where an LMI solution is proposed by considering a nonlinear system that satisfies the Lipschitz condition under the $\mathcal{H}_{-}/\mathcal{H}_{\infty}$ framework.

\subsection{Proposed OL and QB observer design}
The Lipschitz constant is sensitive to the operating region of the system $\mathcal{D}\times\mathcal{U}$ and the parameters that define the nonlinear function $\boldsymbol{\phi}(\cdot)$. In this scenario, \textit{Lemma} \ref{lipschitz-proposition} might presumably fail to obtain the observer's gain matrix $L$ and make the residual generator dynamics stable. Generalizing the Lipchitz condition, by considering the QB and OL conditions, is a suitable alternative according to \cite{Abbaszadeh2010,Zhang2012}.
\begin{definition}\label{one-sided-definition}
The function $\boldsymbol{f}(\cdot)$ is one-sided Lipschitz continuous if a constant $\rho \in \mathbb{R}$ exists such that $\forall\:u\in\mathcal{U}$ and $\forall\:\mb{x},\hat{\mb{x}}\in \mathcal{D}$:
\begin{align}
    \langle\boldsymbol{f}(\mb{x},\mb{u})-\boldsymbol{f}(\hat{\mb{x}},\mb{u}),\mb{x}-\hat{\mb{x}}\rangle \leq \rho \|\mb{x}-\hat{\mb{x}}\|^2. \label{one-sided-lipschitz}
\end{align}
\end{definition}
\begin{definition}\label{qib-definition}
The function $\boldsymbol{f}(\cdot)$ is quadratic inner-bounded continuous if the constants $\delta, \varphi \in \mathbb{R}$ exist such that $\forall\:u\in\mathcal{U}$ and $\forall\:\mb{x},\hat{\mb{x}}\in \mathcal{D}$:
\begin{align}
    \|\boldsymbol{f}(\mb{x},\mb{u})-\boldsymbol{f}(\hat{\mb{x}},\mb{u})\|^2 &\leq \varphi \langle \mb{x}-\hat{\mb{x}}, \boldsymbol{f}(\mb{x},\mb{u})-\boldsymbol{f}(\hat{\mb{x}},\mb{u})\rangle +\notag\\
    &\quad \delta \|\mb{x}-\hat{\mb{x}}\|^2 . \label{quad-bounded}
\end{align}
\end{definition}
\begin{remark}
Different from the Lipschitz constant $\gamma$, the constants $\rho$, $\delta$, and $\varphi$ are not required to be positive. Hence, it is beneficial for a less conservative design of observers.
\end{remark}

In the following theorem, we propose sufficient conditions that incorporate the $\mathcal{H}_{-}/\mathcal{H}_{\infty}$ framework, and conditions \eqref{one-sided-lipschitz} and \eqref{quad-bounded} to demonstrate the existence of the matrix gain $L$.
\setcounter{theorem}{0}
\begin{theorem}\label{one-sided-lipschitz-theorem}
Consider the constants $\rho$, $\delta$, and $\varphi$, the system in \eqref{system-model-with-disturb-faults} satisfying the conditions \eqref{one-sided-lipschitz} and \eqref{quad-bounded}, the observer defined in \eqref{luenberger-observer} and the residual generator in \eqref{error-dynamics}. If exist a filter gain matrix $L$,  strictly positive scalars $\alpha$, $\beta$, and $\{\epsilon_i\}_{i=1}^4$, and identical positive definite matrices $Q$ and $P$, such that the LMIs given as follows hold:
\begin{align}
\begin{split}
\left[\begin{matrix}
    \Omega_1 & P\bar{E}_w + \bar{C}^\top\bar{F}_w & P + \frac{1}{2}(\epsilon_2\varphi - \epsilon_1)I \\
    * & - \alpha^2I + \bar{F}_w^\top\bar{F}_w & 0 \\
    * & * & -\epsilon_2I
    \end{matrix}\right]  
\prec 0 \label{L1}
\end{split}
\end{align}
\begin{align}
\begin{split}
    \left[\begin{matrix}
    \Omega_2 & Q\bar{E}_f - \bar{C}^\top\bar{F}_f & Q + \frac{1}{2}(\epsilon_4\varphi - \epsilon_3)I \\
    * & - \beta^2I + \bar{F}_f^\top\bar{F}_f & 0 \\
    * & * & -\epsilon_4I
\end{matrix}\right] 
\prec 0 \label{L2}
\end{split}
\end{align}
where $\Omega_1 = \bar{A}^\top P + P\bar{A} + \bar{C}^\top\bar{C} + (\epsilon_1\rho + \epsilon_2\delta) I$, $\Omega_2 = \bar{A}^\top Q + Q\bar{A} - \bar{C}^\top\bar{C} + (\epsilon_3\rho + \epsilon_4\delta) I$. Then, we have 
i) the residual generator is stable, and ii)
the observer satisfies the robustness and sensitivity constraints \eqref{rob_cond} and \eqref{sens_cond}.
\end{theorem}
\begin{proof}
Let us choose the following candidate Lyapunov function ${Z} = \mb{e}^\top P\mb{e}$ with $P > 0$ to prove the internal stability of the residual generator \eqref{error-dynamics} in the fault-free case. The time derivative of $V$ is:
\begin{align}
    \dot{Z} &= \mb{e}^\top P\bar{E}_w\mb{w} + \mb{e}^\top(\bar{A}^\top P + P\bar{A})\mb{e} + \mb{w}^\top \bar{E}_w^\top P\mb{e} \notag\\
    &\quad + \mb{e}^\top P\Phi + \Phi^\top P\mb{e}.  \label{Vdot1}
\end{align}
Given $\epsilon_1>0$ and $\epsilon_2>0$, the conditions \eqref{one-sided-lipschitz} and \eqref{quad-bounded} can be expressed equivalently as
\begin{align}
    \Phi^\top\mb{e} \leq \rho\mb{e}^\top\mb{e} &\Leftrightarrow
    \epsilon_1\left(\rho\mb{e}^\top\mb{e} - \Phi^\top\mb{e}\right) \geq 0, \label{one-sided-lipschitz-rewritten1} \\
    \Phi^\top\Phi \leq \varphi \mb{e}^\top\Phi + \delta \mb{e}^\top\mb{e} &\Leftrightarrow
    \epsilon_2(\delta \mb{e}^\top\mb{e} + \varphi \mb{e}^\top\Phi - \Phi^\top\Phi) \geq 0. \label{quad-bounded-rewritten1}
\end{align}

Adding \eqref{one-sided-lipschitz-rewritten1} and \eqref{quad-bounded-rewritten1} to the right-hand side of \eqref{Vdot1}, we have:
\begin{align}
    \begin{split}
        \dot{Z} \leq \mb{w}^\top \bar{E}_w^\top P\mb{e} + \mb{e}^\top(\bar{A}^\top P + P\bar{A})\mb{e} + \mb{e}^\top P\bar{E}_w\mb{w}  \\
        + \Phi^\top P\mb{e} + \mb{e}^\top P\Phi + \epsilon_1 \rho \mb{e}^\top\mb{e} - \epsilon_1\Phi^\top \mb{e} \\
        + \epsilon_2 \delta\:\mb{e}^\top\mb{e} + \epsilon_2 \varphi\:\mb{e}^\top\Phi - \epsilon_2 \Phi^\top\Phi
    \end{split}\label{upper-bound-Vdot-1}
\end{align}
Rewriting \eqref{upper-bound-Vdot-1} as in a matrix inequality:
\begin{align}
    \dot{Z} \leq 
    \left[
        \begin{matrix}
            \mb{e} \\ \mb{w} \\ \mb{\Phi}
        \end{matrix}
    \right]^\top
    \underbrace{
    \scalemath{0.6}{
    \left[
        \begin{matrix}
            (\bar{A}^\top P + P\bar{A}) & P\bar{E}_w & P + \frac{1}{2}(\epsilon_2\varphi - \epsilon_1) I\\
            (\epsilon_1\rho + \epsilon_2\delta) I & & \\
            & & \\
            \bar{E}_w^\top P & 0 & 0 \\
            & & \\
            P + \frac{1}{2}(\epsilon_2\varphi - \epsilon_1) I & 0 & -\epsilon_2I
        \end{matrix}
    \right]}
    }_\text{${\mb{M}}$}
    \left[
        \begin{matrix}
            \mb{e} \\ \mb{w} \\ \mb{\Phi}
        \end{matrix}
    \right] \label{up-bound-without-rob-cond}
\end{align}
Hence, $\dot{Z} < 0$ if $M < 0$. Considering that $Z > 0$ and $\mb{r}_w$ in \eqref{fault_free_res}, we can satisfy constraint \eqref{rob_cond} by rewriting it as:
\begin{align}
    \int_0^t \left(\dot{Z} + \mb{r}_w^\top\mb{r}_w - \alpha^2 \: \mb{w}^\top\mb{w} \right)\: dt  &\leq 0 \notag\\
    \dot{Z} + \mb{r}_w^\top\mb{r}_w - \alpha^2 \: \mb{w}^\top\mb{w} &\leq 0\notag\\
    \mb{e}^\top\bar{C}^\top\bar{C}\mb{e} - \alpha^2 \: \mb{w}^\top\mb{w} + \mb{e}^\top\bar{C}^\top\bar{F}_w\mb{w} & \notag\\
    + \dot{Z} + \mb{w}^\top\bar{F}_w^\top\bar{F}_w\mb{w} + \mb{w}^\top\bar{F}_w\bar{C}\mb{e}  &\leq 0. \label{rob-cond-with-vdot1}
\end{align}
Assuming zero initial conditions, combining \eqref{rob-cond-with-vdot1} with \eqref{up-bound-without-rob-cond} we have:
\begin{align}
    \dot{Z} \leq 
    \left[
        \begin{matrix}
            \mb{e} \\ \mb{w} \\ \mb{\Phi}
        \end{matrix}
    \right]^\top
    \underbrace{
    \scalemath{0.6}{
    \left[
        \begin{matrix}
            (\bar{A}^\top P + P\bar{A}) + \bar{C}^\top\bar{C} & P\bar{E}_w + \bar{C}^\top\bar{F}_w & P + \frac{1}{2}(\epsilon_2\varphi - \epsilon_1) I\\
            (\epsilon_1\rho + \epsilon_2\delta) I & & \\
            & & \\
            \bar{E}_w^\top P + \bar{F}_w\bar{C} & - \alpha^2 I + \bar{F}_w^\top\bar{F}_w & 0 \\
            & & \\
            P + \frac{1}{2}(\epsilon_2\varphi - \epsilon_1) I & 0 & -\epsilon_2I
        \end{matrix}
    \right]}
    }_\text{${\mb{M}^*}$}
    \left[
        \begin{matrix}
            \mb{e} \\ \mb{w} \\ \mb{\Phi}
        \end{matrix}
    \right]
\end{align}
If ${\mb{M}^*} < 0$, then $\dot{Z} < 0$ which means that the residual generator is asymptotically stable in the Lyapunov sense. Therefore, the inequality ${\mb{M}^*} < 0$ is equivalent to \eqref{L1}. 

Similarly, defining the Lyapunov function $Y = \mb{e}^\top Q\mb{e}$ with $Q> 0$ to demonstrate the internal stability of \eqref{error-dynamics} in the disturbance-free case:
\begin{align}
    \begin{split}
    \dot{Y} = \mb{e}^\top(\bar{A}^\top P+P\bar{A})\mb{e} + \mb{e}^\top P\bar{E}_f\mb{f} + \mb{f}^\top \bar{E}_f^\top P\mb{e} \\
    + \mb{e}^\top P\Phi + \Phi^\top P\mb{e}.  
    \end{split}\label{Ydot1}
\end{align}
Adding \eqref{one-sided-lipschitz-rewritten1}, and \eqref{quad-bounded-rewritten1} with strictly positive scalars $\epsilon_3$ and $\epsilon_4$ to the right-hand side of \eqref{Ydot1}:
\begin{align}
    \begin{split}
        \dot{Y} \leq \mb{f}^\top \bar{E}_f^\top P\mb{e} + \mb{e}^\top P\bar{E}_f\mb{f} + \mb{e}^\top(\bar{A}^\top P + P\bar{A})\mb{e}  \\
        + \Phi^\top P\mb{e} + \epsilon_3 \rho \mb{e}^\top\mb{e} + \mb{e}^\top P\Phi  - \epsilon_3 \Phi^\top \mb{e} \\
        + \epsilon_4 \delta\:\mb{e}^\top\mb{e} + \epsilon_4 \varphi\:\mb{e}^\top\Phi - \epsilon_4 \Phi^\top\Phi.
    \end{split}\label{upper-bound-Ydot}
\end{align}
Conducting a similar derivation for constraint \eqref{sens_cond} with Schur complements and elementary manipulations yields the matrix inequality \eqref{L2}.
\end{proof}
Notice that matrices in \eqref{L1} and \eqref{L2} are nonlinear in terms of $L$, $P$ and $Q$. However, if we set $Y=PL$, then \eqref{L1} and \eqref{L2} become LMIs in $Y$. The matrix gain $L$ can be obtained using $L=P^{-1}Y$, once the problem is solved. Hereafter, we refer to our proposed observer design as the OL observer or OL-QB observer.

\begin{figure}
\centering
  \begin{subfigure}{0.23\textwidth}
    \centering
    \includegraphics[width=1.0\linewidth]{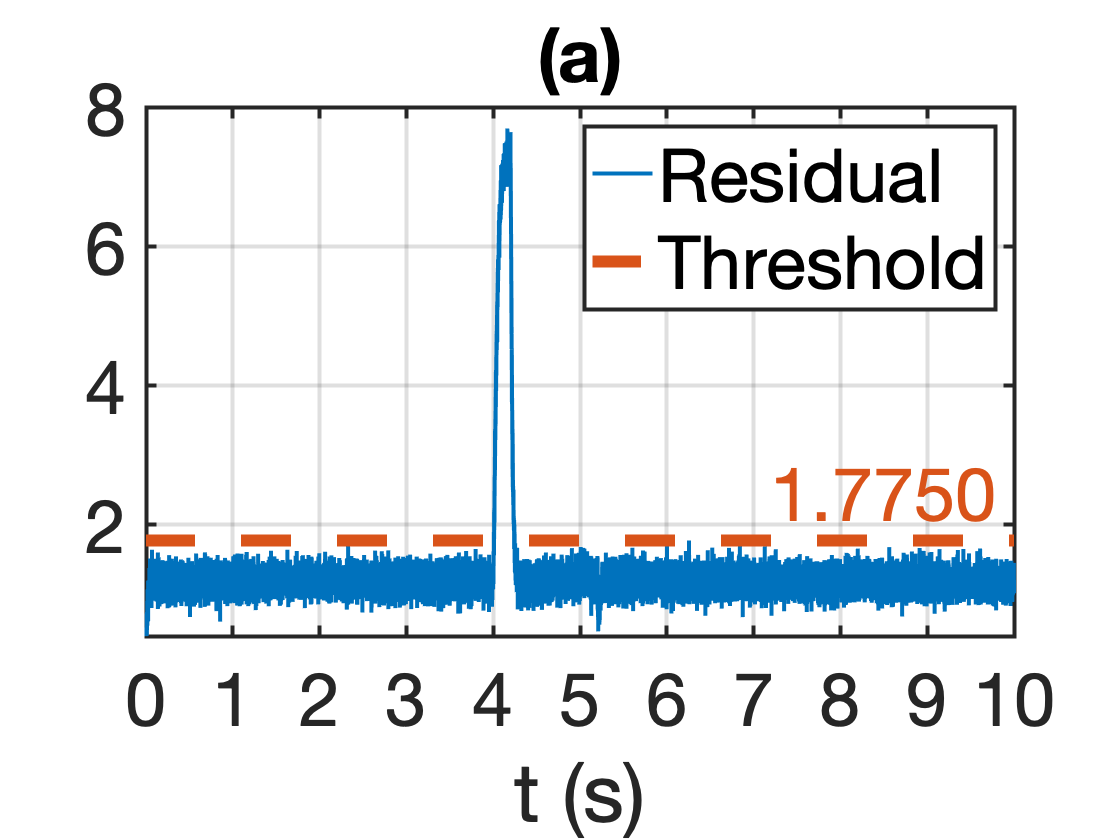}
  \end{subfigure}%
  \begin{subfigure}{0.23\textwidth}
    \centering
    \includegraphics[width=1.0\linewidth]{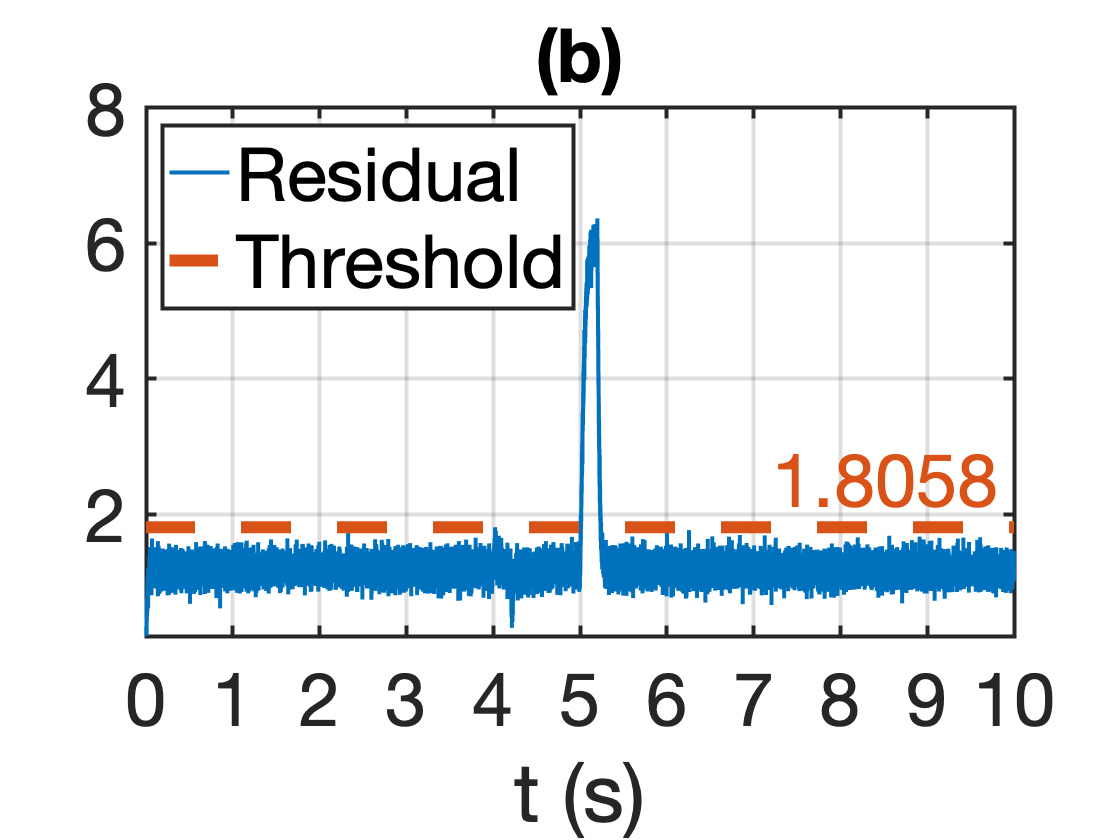}
  \end{subfigure}%
  
  \begin{subfigure}{0.23\textwidth}
    \centering
    \includegraphics[width=1.0\linewidth]{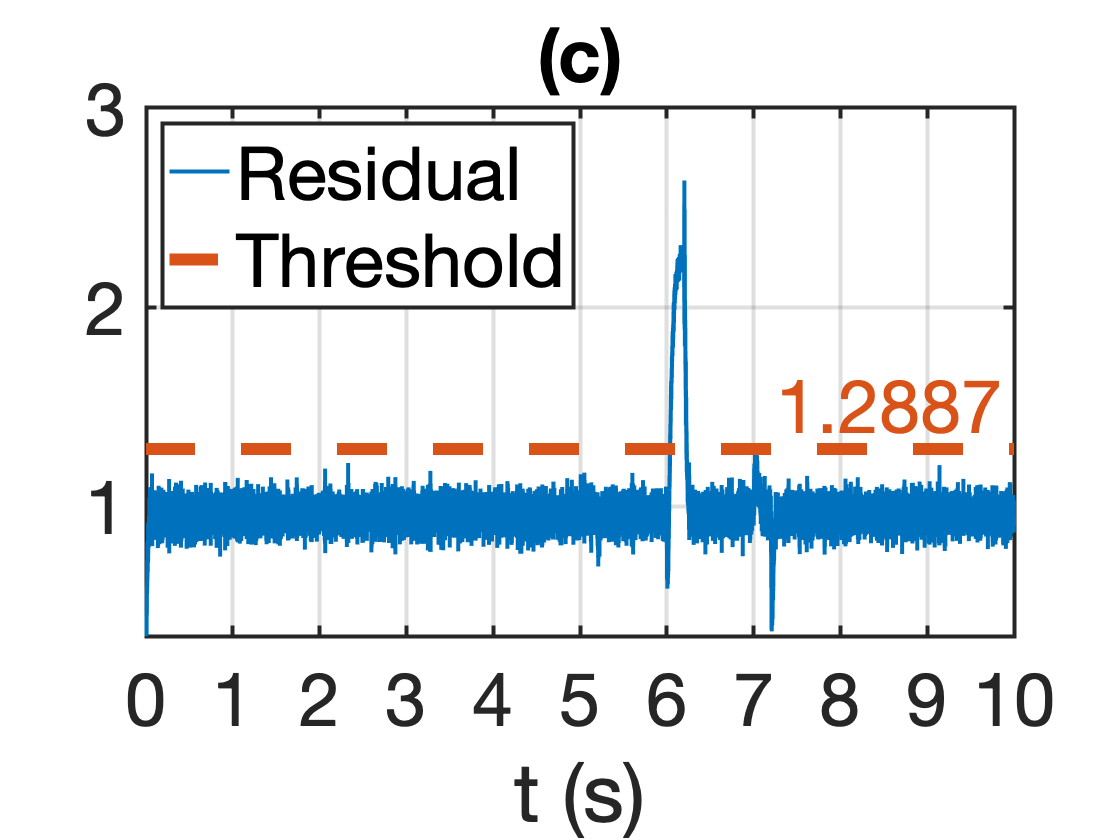}
  \end{subfigure}%
  \begin{subfigure}{0.23\textwidth}
    \centering
    \includegraphics[width=1.0\linewidth]{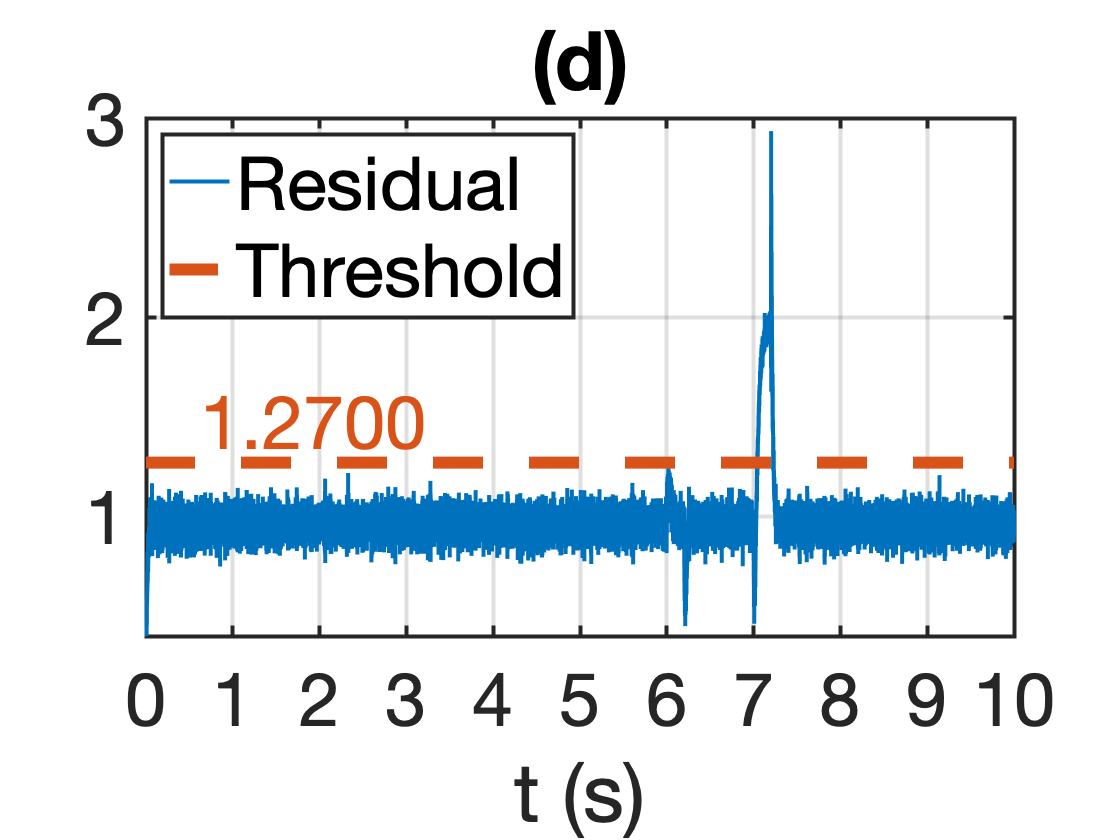}
  \end{subfigure}%
\caption{Residual norms under busbar fault at the PCC using a one-sided Lipschitz observer: (a) GFM $\#1$ at $t=4$, (b) GFM $\#2$ at $t=5$, (c) GFM $\#3$ at $t=6$, and (d) GFM $\#4$ at $t=7$. The faults are cleared after $0.2$ seconds.} 
\label{one-sided-lipschitz-3ph}
\vspace{-0.2cm}
\end{figure}

\begin{figure}
\centering
  \begin{subfigure}{0.23\textwidth}
    \centering
    \includegraphics[width=1.0\linewidth]{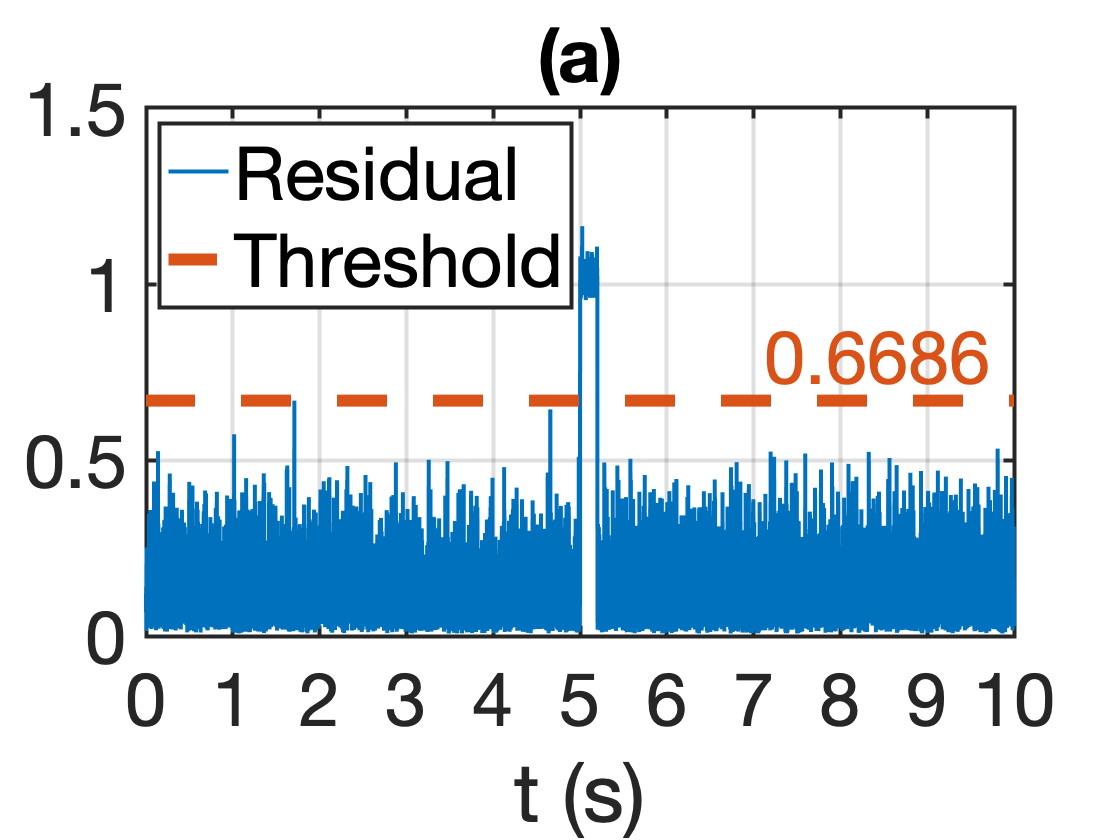}
  \end{subfigure}%
  \begin{subfigure}{0.23\textwidth}
    \centering
    \includegraphics[width=1.0\linewidth]{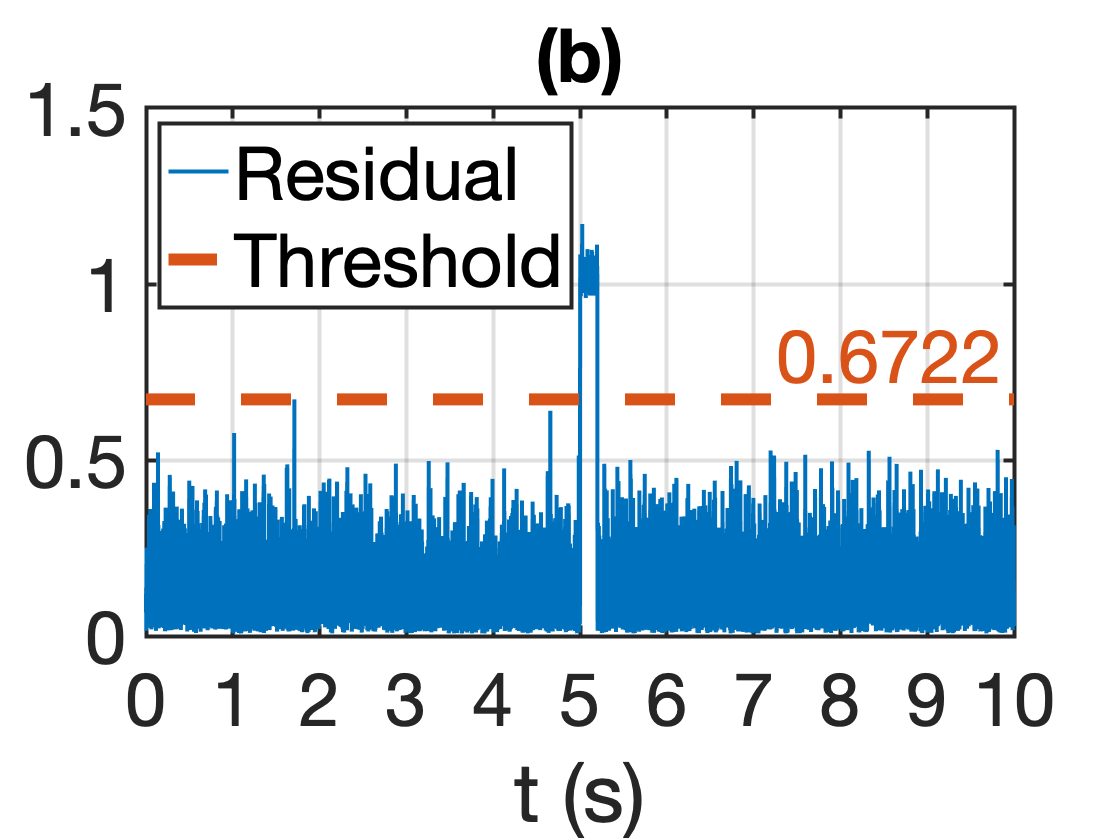}
  \end{subfigure}%
  
  \begin{subfigure}{0.23\textwidth}
    \centering
    \includegraphics[width=1.0\linewidth]{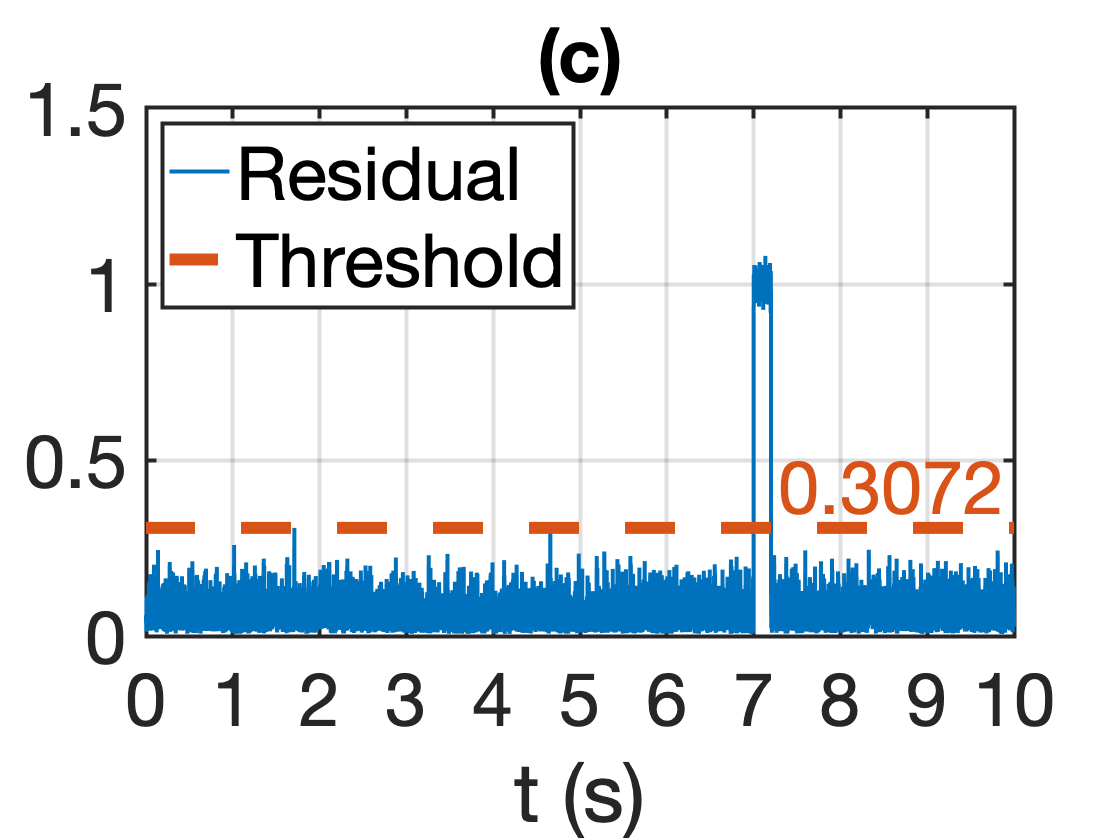}
  \end{subfigure}%
  \begin{subfigure}{0.23\textwidth}
    \centering
    \includegraphics[width=1.0\linewidth]{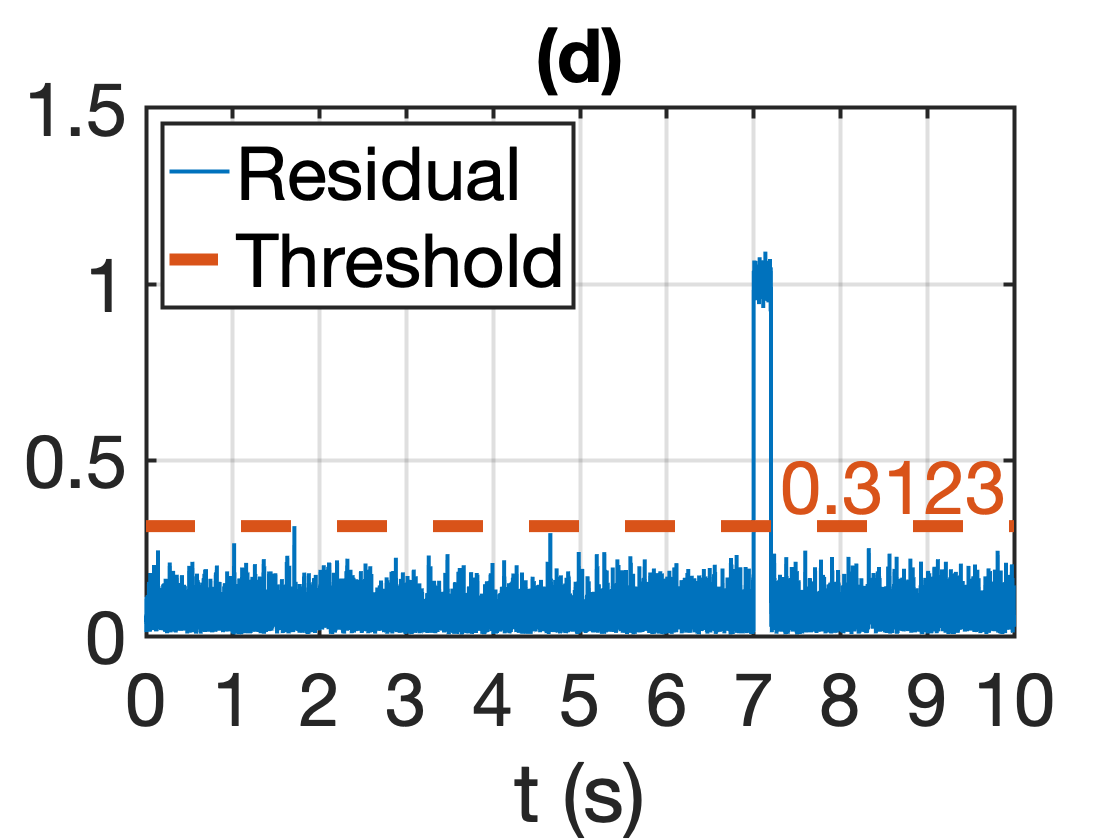}
  \end{subfigure}%
\caption{Residual norms under an actuator fault corresponding to input $\omega_{ni}$ using a one-sided Lipschitz observer (left column) and Lipschitz observer (right column): (a) and (b) GFM $\#2$ at $t=5$; (c) and (d) GFM $\#4$ at $t=7$. The fault is cleared after $0.2$ seconds.} 
\label{wni}
\vspace{-0.5cm}
\end{figure}

\begin{figure}
\centering
  \begin{subfigure}{0.23\textwidth}
    \centering
    \includegraphics[width=1.0\linewidth]{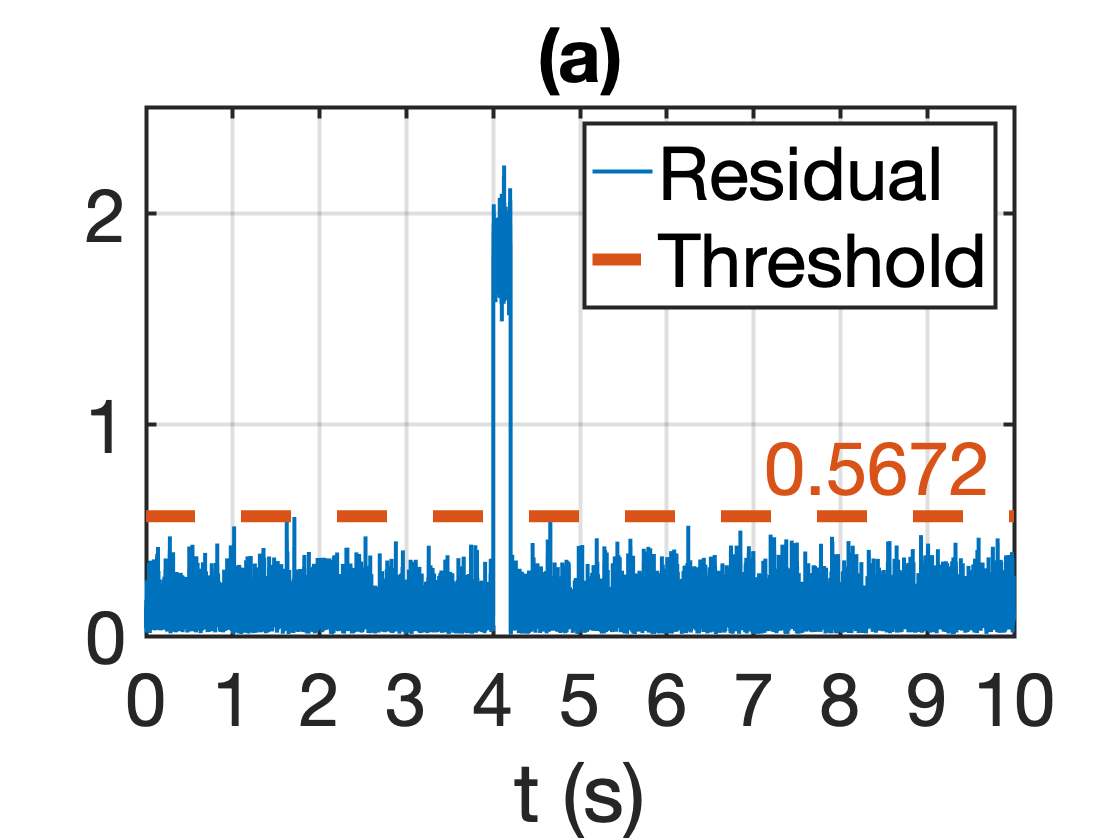}
  \end{subfigure}%
  \begin{subfigure}{0.23\textwidth}
    \centering
    \includegraphics[width=1.0\linewidth]{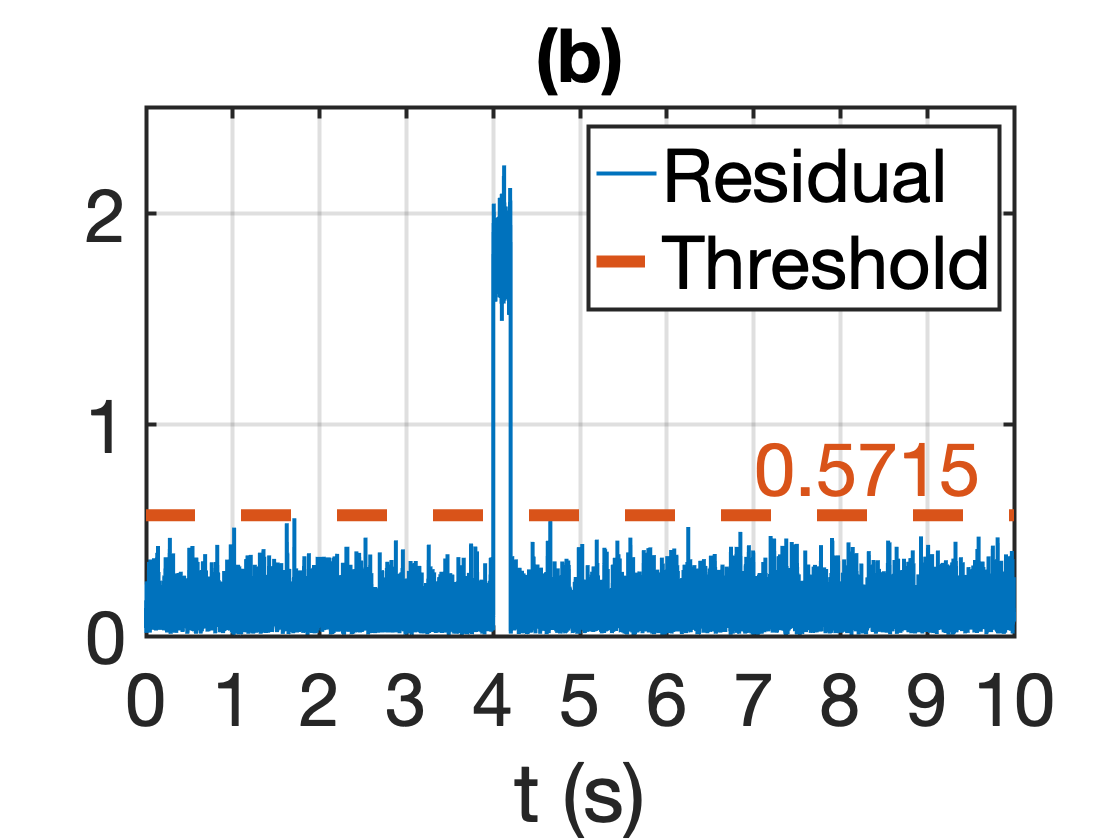}
  \end{subfigure}%
  
  \begin{subfigure}{0.23\textwidth}
    \centering
    \includegraphics[width=1.0\linewidth]{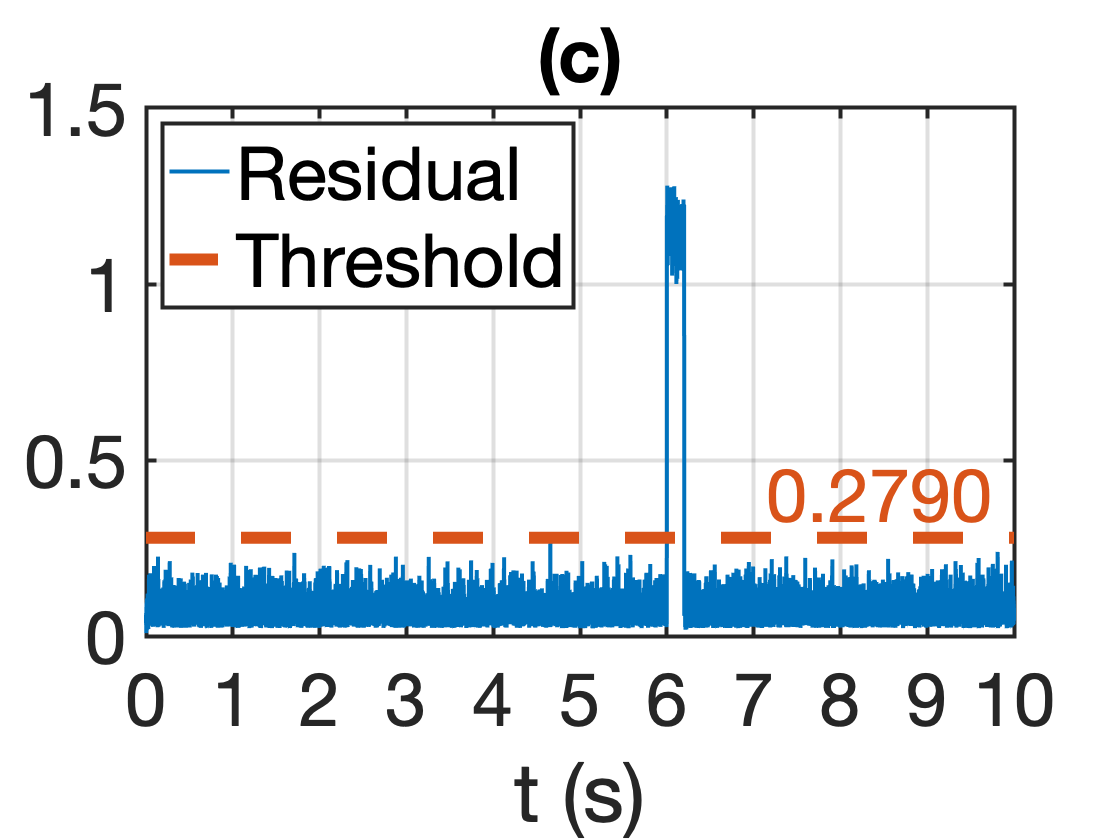}
  \end{subfigure}%
  \begin{subfigure}{0.23\textwidth}
    \centering
    \includegraphics[width=1.0\linewidth]{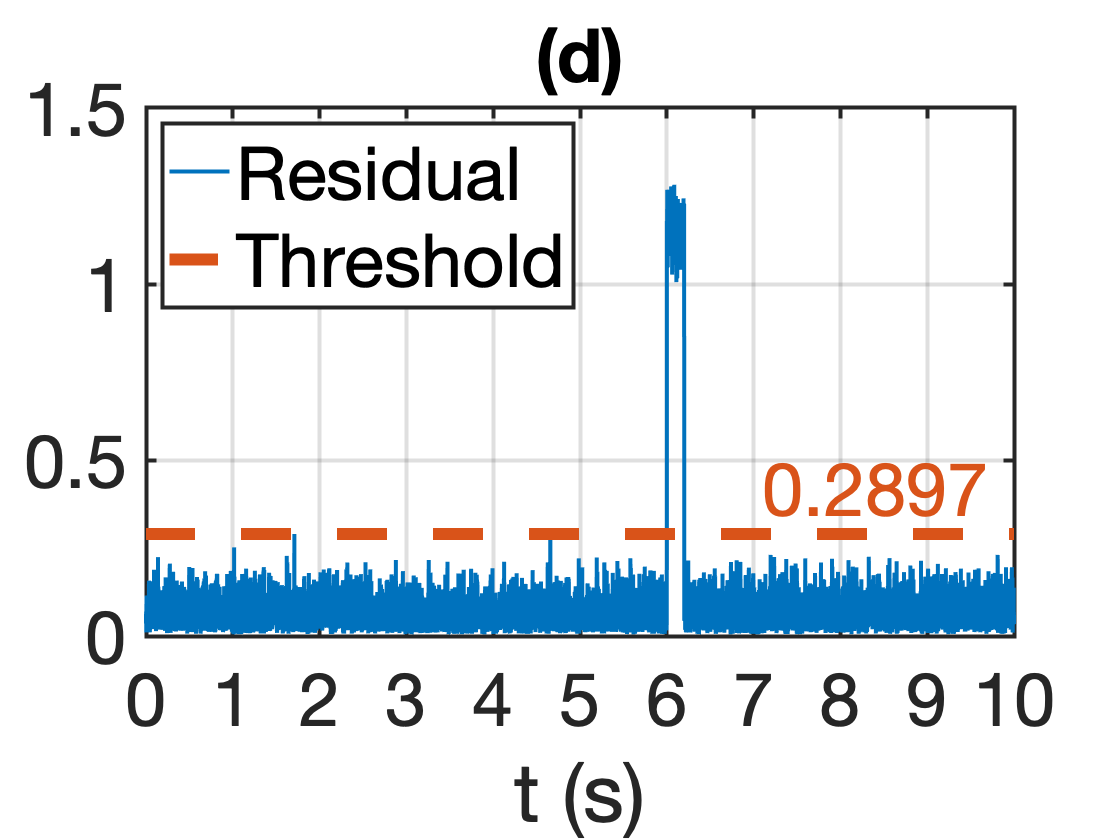}
  \end{subfigure}%
\caption{Residual norms under an actuator fault corresponding to input $V_{ni}$ using a one-sided Lipschitz observer (left column) and Lipschitz observer (right column): (a) and (b) GFM $\#1$ at $t=4$; (c) and (d) GFM $\#3$ at $t=6$. The fault is cleared after $0.2$ seconds.} 
\label{Vni}
\vspace{-0.3cm}
\end{figure}

\begin{figure}
\centering
  \begin{subfigure}{0.23\textwidth}
    \centering
    \includegraphics[width=1.0\linewidth]{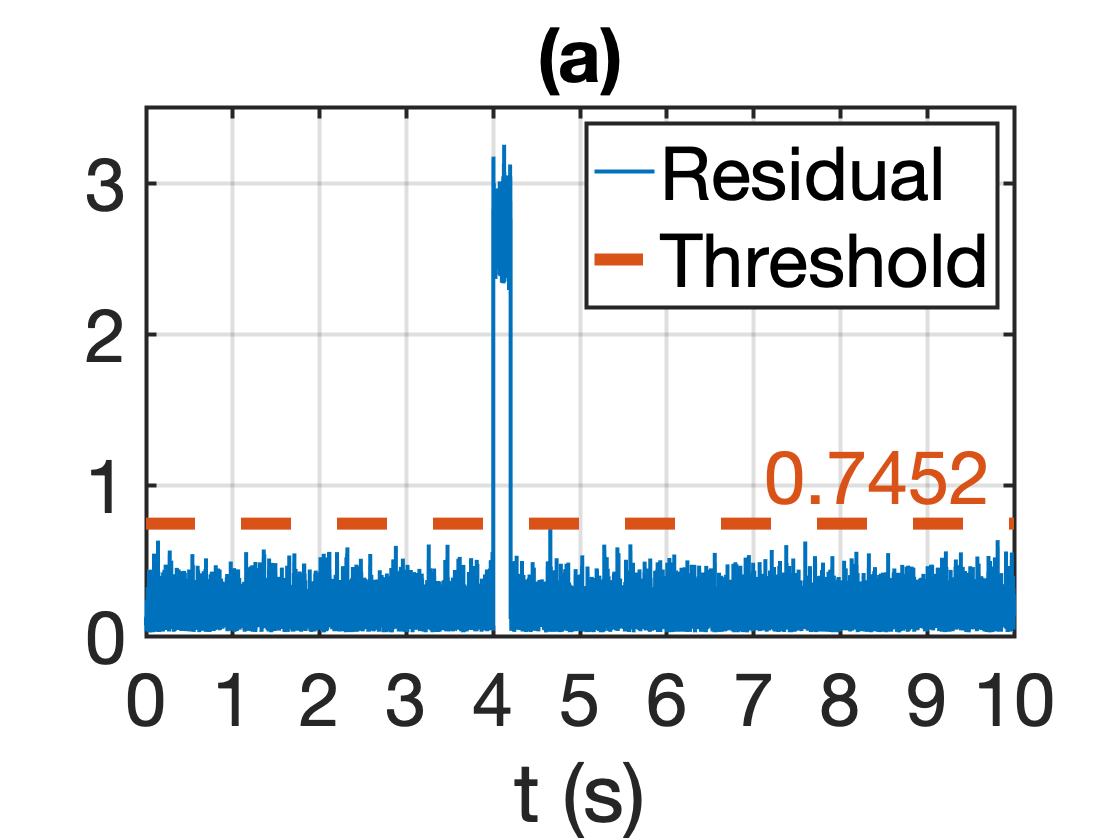}
  \end{subfigure}%
  \begin{subfigure}{0.23\textwidth}
    \centering
    \includegraphics[width=1.0\linewidth]{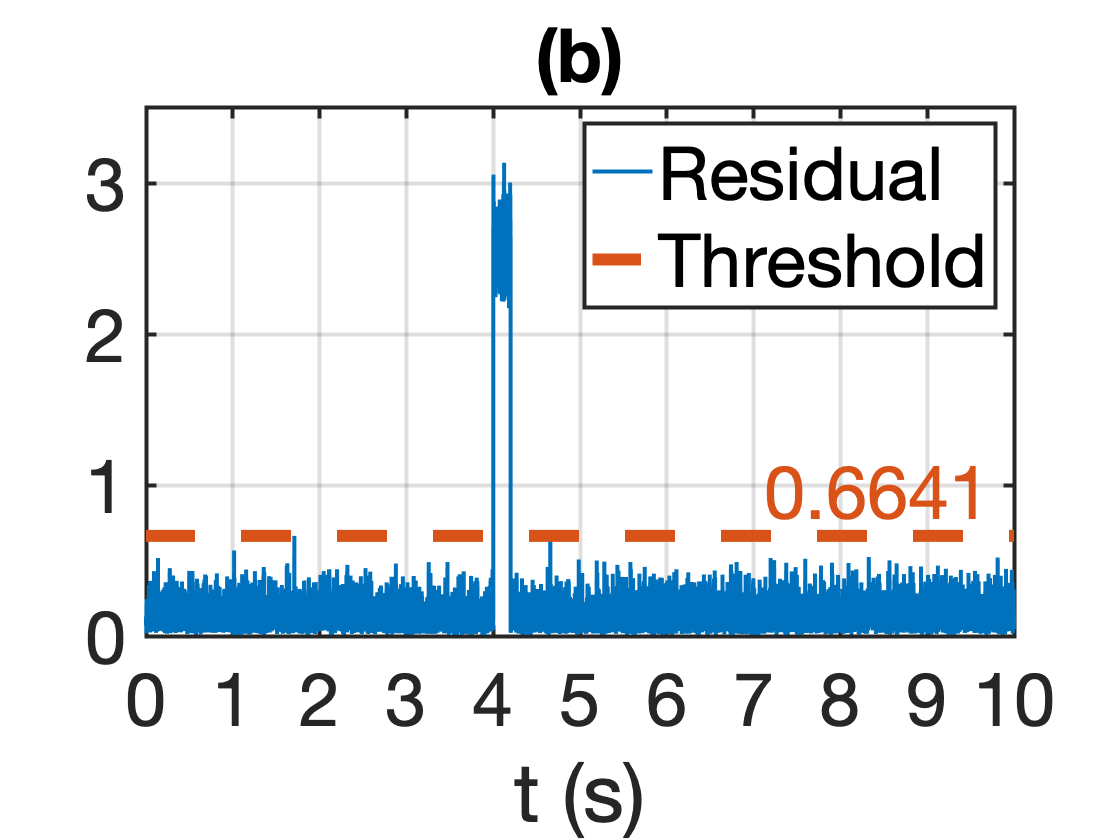}
  \end{subfigure}%
  
  \begin{subfigure}{0.23\textwidth}
    \centering
    \includegraphics[width=1.0\linewidth]{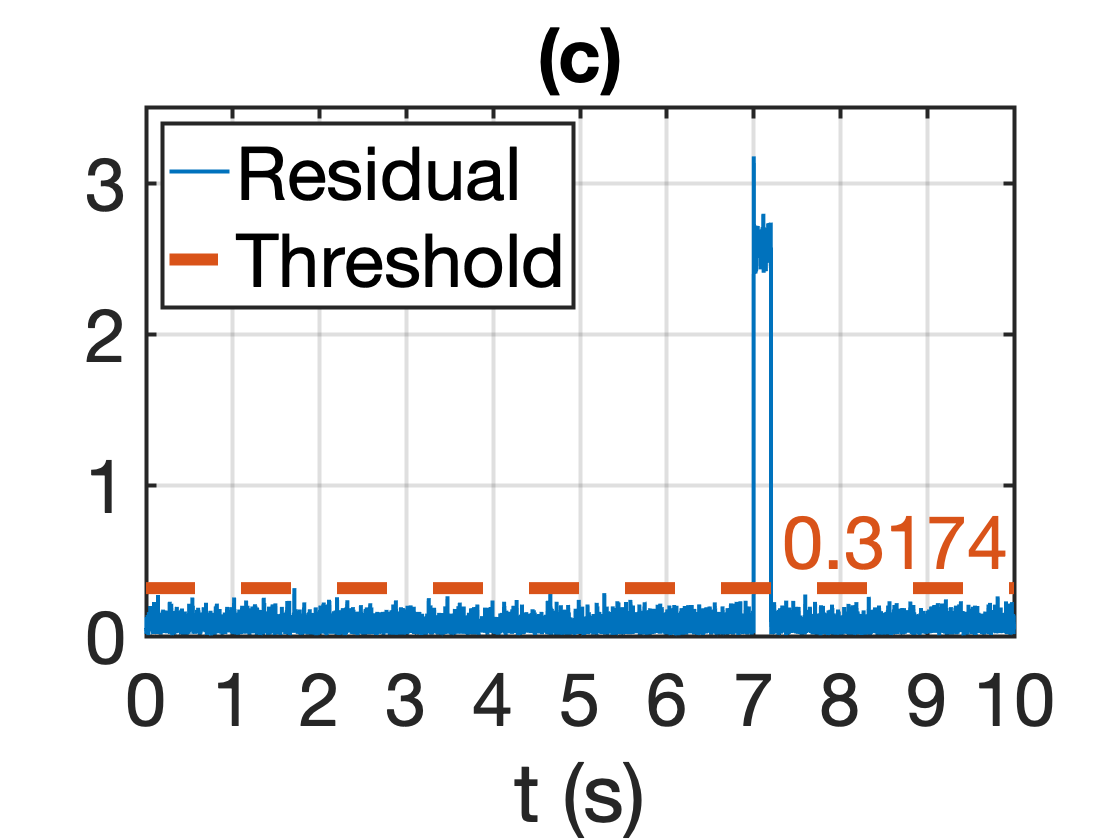}
  \end{subfigure}%
  \begin{subfigure}{0.23\textwidth}
    \centering
    \includegraphics[width=1.0\linewidth]{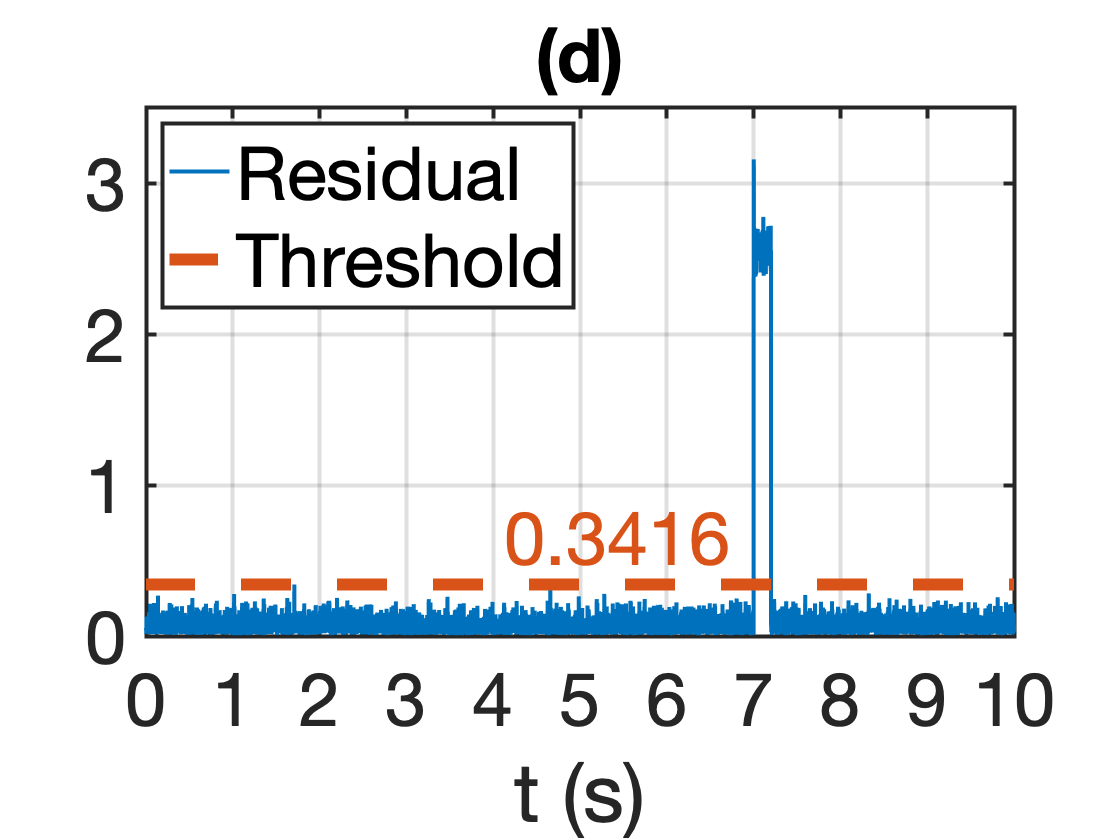}
  \end{subfigure}%
\caption{Residual norms under an inverter bridge fault (sudden efficiency reduction) using a one-sided Lipschitz observer (left column) and Lipschitz observer (right column): (a) and (b) GFM $\#1$ at $t=4$; (c) and (d) GFM $\#4$ at $t=7$. The fault is cleared after $0.2$ seconds.} 
\label{bridge}
\vspace{-0.5cm}
\end{figure}

\begin{figure}
\centering
\begin{tikzpicture}[font=\scriptsize,node distance=0.5cm,thick]
\node[draw,
    rounded rectangle,
    minimum width=1.5cm,
    minimum height=0.5cm] (block1) {START};
 
\node[draw,
    below=of block1,
    minimum width=3.5cm,
    minimum height=0.5cm,
    align=center
] (block2) {Compute the OL and QB constants according to \\ Definitions \ref{one-sided-definition} and \ref{qib-definition}};

\node[draw,
    below=of block2,
    minimum width=3.5cm,
    minimum height=0.5cm
] (block3) {Select the fault vector $\mb{f}$ and matrix $E_f$ according to Section \ref{inverter-model}};

\node[draw,
    below=of block3,
    minimum width=3.5cm,
    minimum height=0.5cm,
    align=center
] (block4) {Obtain the observer gain matrix $L$ by solving \\ the feasibility program  presented in Theorem \ref{one-sided-lipschitz-theorem}};

\node[draw,
    align=center,
    below=of block4,
    minimum width=3.5cm,
    minimum height=0.5cm
] (block5) {Compute the threshold $J_{th}$ according to \eqref{residual-threshold}};

\node[draw,
    align=center,
    below=of block5, 
    minimum width=3.5cm,
    minimum height=0.5cm
] (block6) {Compute the residual norm $\mb{r}$ according to \eqref{error-dynamics-2}};

\node[draw,
    diamond,
    below=of block6,
    minimum width=2.0cm,
    inner sep=0] (block7) {$\|\mb{r}\|_2 \leq J_{th}$};

\node[draw,
    align=center,
    below=of block7,
    minimum width=3.5cm,
    minimum height=0.5cm
] (block8) {Trigger fault alarm};

\node[draw,
    diamond,
    below=of block8,
    minimum width=2.5cm,
    inner sep=0] (block9) {Is the fault cleared?};

\draw[-latex] (block1) edge (block2)
    (block2) edge (block3)
    (block3) edge (block4)
    (block4) edge (block5)
    (block5) edge (block6)
    (block6) edge (block7)
    (block8) edge (block9);

\draw[-latex] 
    (block7) edge node[pos=0.4,fill=white,inner sep=2pt]{No}(block8)
    (block7) -| ++(-3.55,0) node[pos=0.1,fill=white,inner sep=0]{Yes} |- (block6);

\draw[-latex] 
    (block9) -| ++(-3,0) node[pos=0.2,fill=white,inner sep=2pt]{No} |- (block8);

\draw[-latex] 
    (block9.east) -| ++(1.8,0) node[pos=0.2,fill=white,inner sep=0]{Yes} |- (block6.east);

\end{tikzpicture}
\caption{Flowchart of the proposed fault detection algorithm.}
\label{flowchart}
\vspace{-0.4cm}
\end{figure}

\subsection{Relation between Lipschitz and one-sided Lipschitz observer design for nonlinear systems}
The relation between {Lemma \ref{lipschitz-proposition}}  and {Theorem \ref{one-sided-lipschitz-theorem}} is established in {Theorem \ref{relationship}}, which indicates that the latter is less conservative.
\begin{theorem}\label{relationship}
Assume the function $\boldsymbol{\phi}(\mb{x},\mb{u})$ is nonlinear and Lipschitz continuous with $\gamma$ as its Lipschitz constant, and there exist the gain matrices $L$, $P$, $Q$, non-negative scalars $\varepsilon_1$, $\varepsilon_2$ such that the inequalities \eqref{L1-lipschitz} and \eqref{L2-lipschitz} hold. Then, there exist non-negative scalars $\epsilon_1$, $\epsilon_2$, $\epsilon_3$, $\epsilon_4$, real scalars $\rho$, $\delta$, $\varphi$, together with the matrices $L$, $P$, and $Q$ such that the inequalities in \eqref{L1} and \eqref{L2} are satisfied.
\end{theorem}
\begin{proof}
The Lipschitz condition implies the OL and QB conditions \cite{Abbaszadeh2010}. Using the Cauchy-Schwarz inequality and the assumption that $\boldsymbol{\phi}(\mb{x},\mb{u})$ is nonlinear satisfying the Lipschitz condition in $\mathcal{U}\times\mathcal{D}$, we obtain:
\begin{align*}
    \lvert\langle\hat{\mb{x}}-\mb{x},\boldsymbol{\phi}(\hat{\mb{x}},\mb{u})-\boldsymbol{\phi}(\mb{x},\mb{u})\rangle\rvert &\leq \|\hat{\mb{x}}-\mb{x}\| \|\boldsymbol{\phi}(\hat{\mb{x}},\mb{u})-\boldsymbol{\phi}(\mb{x},\mb{u})\|  \notag \\
    & \leq \gamma \|\hat{\mb{x}}-\mb{x}\|^2. 
\end{align*}
The OL constant $\rho$ can be set equal to $\pm \gamma$ based on Definition \ref{one-sided-definition}. Similarly, using Definition \ref{qib-definition} we get:
\begin{align*}
    \|\boldsymbol{\phi}(\hat{\mb{x}},\mb{u})-\boldsymbol{\phi}(\mb{x},\mb{u})\|^2 &\leq \varphi \langle \boldsymbol{\phi}(\mb{x},\mb{u})-\boldsymbol{\phi}(\hat{\mb{x}},\mb{u}), \mb{x}-\hat{\mb{x}} \rangle \notag \\
    &\quad + \delta \|\hat{\mb{x}}-\mb{x}\|^2 \notag \\
    &\leq \left(\delta + \varphi\rho \right)\|\mb{x}-\hat{\mb{x}}\|^2\,.
\end{align*}

We can obtain the value of $\delta$ by setting
$\gamma^2 = \left(\delta + \varphi\rho \right)$. Notice that the previous inequality is valid if $\varphi > 0$. By comparing \eqref{L1} with \eqref{L1-lipschitz}, we establish the following relations:
$\epsilon_1 = \varepsilon_1\varphi$,
$\epsilon_2 = \varepsilon_1$, and
$\gamma^2 = \delta + \varphi\rho$.
Given that $\varphi > 0$, then $\epsilon_1 > 0$. A similar analysis can be done by using the entries of inequalities \eqref{L2} and \eqref{L2-lipschitz}.
\end{proof}
\begin{remark}
One-sided Lipschitz continuity does not imply Lipschitz continuity, because the Lipschitz condition is a two-sided inequality \mbox{\cite{Abbaszadeh2010}}. Hence, the converse of Theorem {\ref{relationship}} does not hold, which implies that Theorem {\ref{one-sided-lipschitz-theorem}} is less conservative than Lemma {\ref{lipschitz-proposition}}.
\end{remark}

Figure~{\ref{flowchart}} presents the flowchart of our proposed fault detection algorithm. The flowchart succinctly summarizes the main steps and illustrates the details of the proposed approach.

\section{Numerical tests}\label{numerical-tests}
We evaluate our proposed observer- and residual-based strategy using the islanded AC microgrid shown in Fig.~\ref{test-system}. The microgrid has four GFMs, three inductive branches, and four RL loads. The parameters of the microgrid, GFMs, branches, and loads are presented in Tables \ref{parameters-gfms} and \ref{parameters-microgrid}, as described in \cite{Bidram2013}. The experiments were performed using MATLAB, Simulink, and the \textit{YALMIP} toolbox with the semidefinite programming solvers \textit{LMILAB}, \textit{MOSEK}, \textit{SeDuMi}, and \textit{SDPT3} to solve optimization problem described by {Lemma} \ref{lipschitz-proposition} and {Theorem} \ref{one-sided-lipschitz-theorem}. The Lipschitz, OL, and QB constants are presented in Table \ref{constants} and computed according to \cite{Qi2018}. For a single type of fault, the fault simulations happen in this order: at $t=4$ for GFM \#1, at $t=5$ for GFM \#2, at $t=6$ for GFM \#3, and at $t=7$ for GFM \#4. All the faults are cleared after $0.2$ seconds. The disturbances are modeled as uncorrelated white Gaussian noise on the process and measurements equations \cite{Shoaib2022}. The disturbance matrices are $E_w = B$, and $F_w = D$ \cite{Shoaib2022}. The faults of other GFMs are also considered disturbances or unknown inputs. For each of our experiments, we use a window of ten seconds to find the highest value of the residual vector $\mb{r}$. Such a value is selected as the corresponding threshold $J_{th}$ at each experiment. The following simulations compare our proposed approach based on the OL and QB conditions against the state-of-the-art Lipschitz design. 


\begin{table}[t]
\caption{Parameters of the grid-forming inverters.}
\begin{center}
\begin{tabular}{c|cc} 
\hline
Parameter           & \multicolumn{2}{c}{Values}                  \\
                    & GFM \#1 {\&} \#2          & GFM \#3 {\&} \#4        \\
\hline
Power rating (kVA)  & 45                    & 34                  \\
Voltage rating (V)  & 380                   & 380                 \\
$m_{Pi}$            & $9.4\times10^{-5}$    & $12.5\times10^{-5}$ \\
$n_{Qi}$	        & $1.3\times10^{-3}$    & $1.5\times10^{-3}$  \\
$R_{ci}\:(\Omega)$ 	& $0.03$                & $0.03$              \\
$L_{ci}\:(mH)$	    & $0.35$                & $0.35$              \\
$R_{fi}\:(\Omega)$  & $0.1$                 & $0.1$               \\ 
$L_{fi}\:(mH)$      & $1.35$                & $1.35$              \\
$C_{fi}\:(\:\mu F)$	& $50$                  & $50$                \\
$K_{PVi}$ 	        & $0.1$                 & $0.05$              \\
$K_{IVi}$ 	        & $420$                 & $390$               \\
$K_{PCi}$ 	        & $15$                  & $10.5$              \\    
$K_{ICi}$           & $20000$               & $16000$             \\
$\omega_b\:(rad/s)$	& $314.16$              & $314.16$            \\
\hline
\end{tabular}
\end{center}
\label{parameters-gfms}
\vspace{-0.2cm}
\end{table}

\begin{table}[t]
\caption{Parameters of the AC microgrid.}
\begin{center}
\begin{tabular}{c|cc} 
\hline
              & \multicolumn{2}{c}{Values}  \\
              & Resistance $(\Omega)$   & Inductance $(\mu H)$ \\
\hline   				
Line 1        & $0.23$                  & $318$ \\
Line 2	      & $0.35$                  & $1847$ \\
Line 3 	      & $0.23$                  & $318$ \\
Load 1 	      & $30$                    & $0.477$ \\
Load 2        & $20$                    & $0.318$ \\ 
Load 3        & $25$                    & $0.318$  \\
Load 4	      & $25$                    & $0.477$  \\
\hline
\end{tabular}
\end{center}
\label{parameters-microgrid}
\vspace{-0.5cm}
\end{table}

\subsection{Threshold and residual computation}\label{threshold-residual-computation}
The threshold is computed over a finite-length time period according to two criteria. The first criterion points out that the longer the time period, the better the estimate of the true value of {\eqref{residual-threshold}}. The second criterion suggests that a more extended time period leads to selecting a threshold that helps reduce the rate of false alarms. In this regard, we recommend choosing the minimum data window to be at least ten times the duration of the faults. Although we make the faults last 0.2 s in this work, we have selected a more conservative data window of 10 s  that ensures satisfying the aforementioned criteria. We have selected a more conservative data window of 10 s for threshold computation that ensures satisfying the aforementioned criteria. Notice that the computed threshold is valid before, during, and after the occurrence of the fault, which leads to not requiring a minimum post-fault data. Besides the threshold, the residual signal computation does not require a data window or minimum post-fault data. Instead, the residual calculation is instantaneous because the residual signal is theoretically obtained from the error dynamics represented by {\eqref{error-dynamics-2}}. Notice that the residual is calculated online before, during, and after the presence of the faults.


\subsection{Droop control}\label{droop-control}
\subsubsection{Busbar fault}
Fig.~\ref{one-sided-lipschitz-3ph} shows the response of the residual norm using an OL-QB observer when a busbar fault occurs at the end of the output connector of all GFMs. The fault impedance considered is 10\% of the nominal voltage magnitude of the busbar. The depicted figure illustrates the sensitivity of the four residuals to the busbar fault, with all residuals measuring less than two units. An important observation is that the residual for an individual GFM remains robust against disturbances and faults occurring at other GFMs' PCC. The fault location can be readily identified due to the distinct sensitivity of each residual to its corresponding busbar fault. The observer accommodates the faults at other busbars as disturbances. GFM \#1 exhibits the smallest detection time of $t_o = 8$ ms, whereas GFM \#4 exhibits the largest around $t_o = 45$ ms. The lowest clearing time $t_c = 30$ ms corresponds to GFM \#3, whereas GFM \#1 exhibits the highest $t_c = 49$ ms.

\subsubsection{Actuator faults}
Actuator faults considering the input signal $\omega_{ni}$ for GFMs \#2 and \#4, and the input signal $V_{ni}$ for GFMs \#1 and \#3 are shown in Figs. \ref{wni} and \ref{Vni} respectively. We notice that the behavior of the residual norms is very similar for both the one-sided Lipschitz and Lipschitz observers. The residual thresholds are a little less for the one-sided Lipschitz observer. Notice that the residual norm of each GFM remains unaffected by the occurrence of the faults in the other GFMs. In other words, the actuator faults are correctly identified among the GFMs. Both observer types present fault detection and clearing times of less than 1 ms.

\subsubsection{Inverter bridge fault}
A similar situation occurs with the inverter bridge fault, as shown in Fig.~\ref{bridge}. The inverter bridge fault corresponds to a reduction in efficiency of the bridge. The residual norm is below the threshold during the absence of the fault, and it stays above the threshold under the presence of the fault. The location of the bridge faults is properly identified because the residual norms are sensitive to the corresponding GFM. Both observers present fault detection and clearing times of less than 1 ms.

\begin{table}[t]
\caption{Lipschitz, OL and QB constants for all the grid-forming inverters.}
\begin{center}
\begin{tabular}{c|cccc} 
\hline
GFM \#       & $\gamma$   & $\rho$      & $\delta$    & $\varphi$ \\
\hline
1 \& 2   & 44.7488    & 22.3688     & -0.7493     & 2.3599 \\
3 \& 4   & 44.7488    & 22.3688     & -0.7535     & 2.3679 \\
\hline
\end{tabular}
\end{center}
\label{constants}
\vspace{-0.2cm}
\end{table}

\begin{table}[t]
\caption{Computational time of OL-QB and Lipschitz observers: Five-fold average values of the mean $\mu$ and the standard deviation $\sigma$ (both in seconds).}
\begin{center}
\begin{tabular}{c|cc|cc} 
\hline
Fault type    & \multicolumn{2}{c|}{OL-QB} & \multicolumn{2}{c}{Lipschitz} \\
       & $\mu$      & $\sigma$       & $\mu$         & $\sigma$ \\
\hline
Three-phase         & 1144      & 7.77        & --         & -- \\ 
$\omega_{ni}$       & 1736      & 3.85        & 4076       & 14.89 \\
$V_{ni}$            & 1472      & 3.11        & 3784     & 12.76 \\
Inverter bridge     & 1686    & 6.91     & 1916     & 6.02 \\
\hline
\end{tabular}
\end{center}
\label{processing-time}
\vspace{-0.7cm}
\end{table}


\subsection{Virtual synchronous machine control}\label{vsm}
\begin{figure}
\centering
  \begin{subfigure}{0.23\textwidth}
    \centering
    \includegraphics[width=1.0\linewidth]{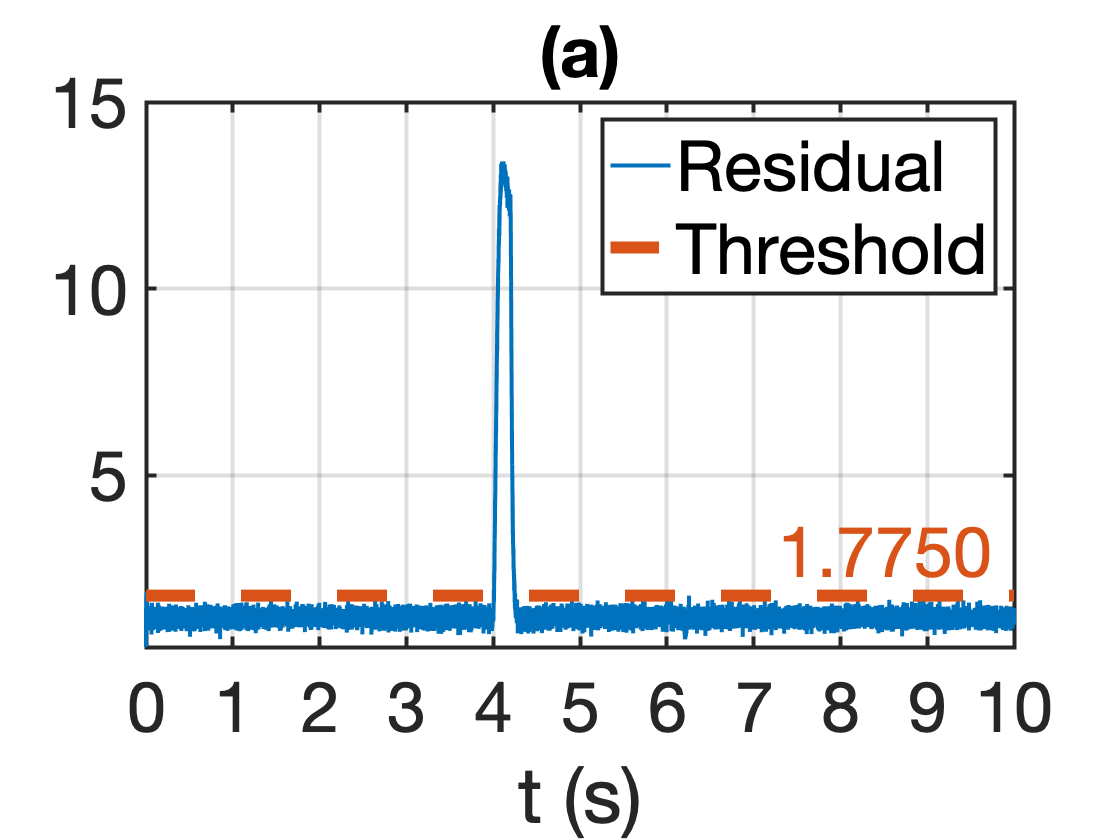}
    \label{diverse-intriago-khan-(a)}
  \end{subfigure}%
  \begin{subfigure}{0.23\textwidth}
    \centering
    \includegraphics[width=1.0\linewidth]{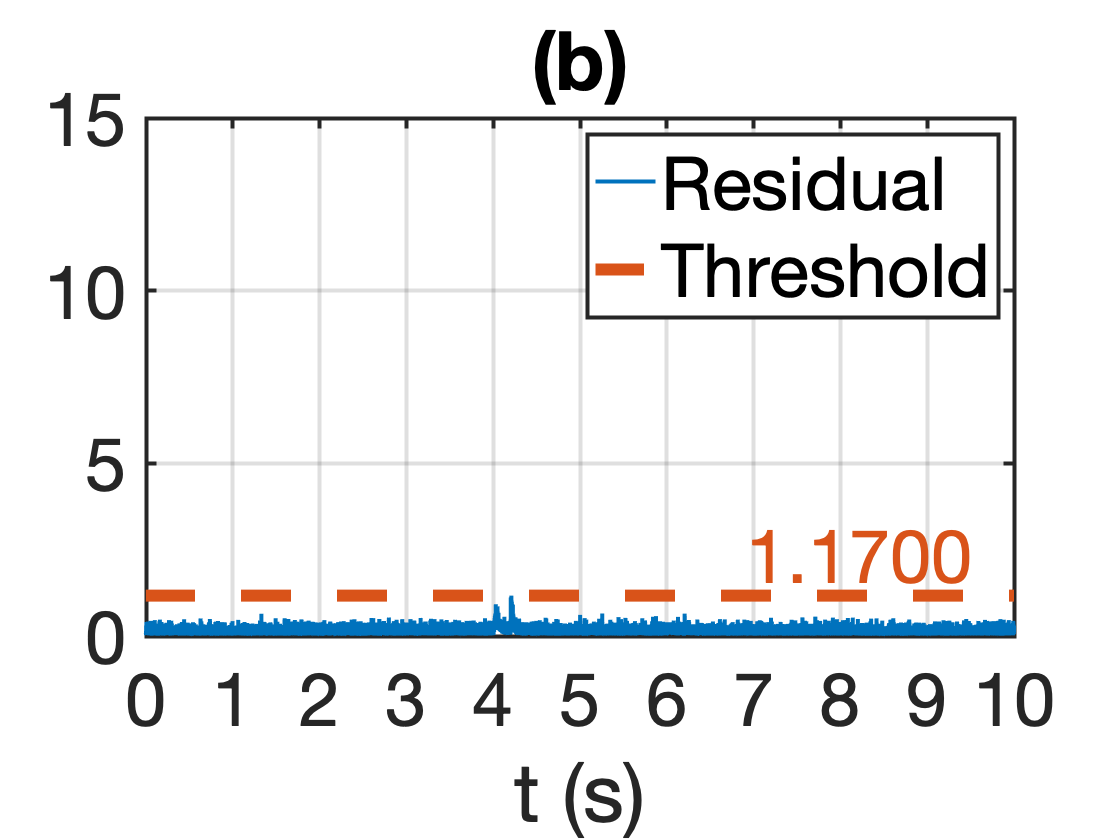}
    \label{diverse-intriago-khan-(b)}
  \end{subfigure}%
  \vspace{-0.3cm}
  \begin{subfigure}{0.23\textwidth}
    \centering
    \includegraphics[width=1.0\linewidth]{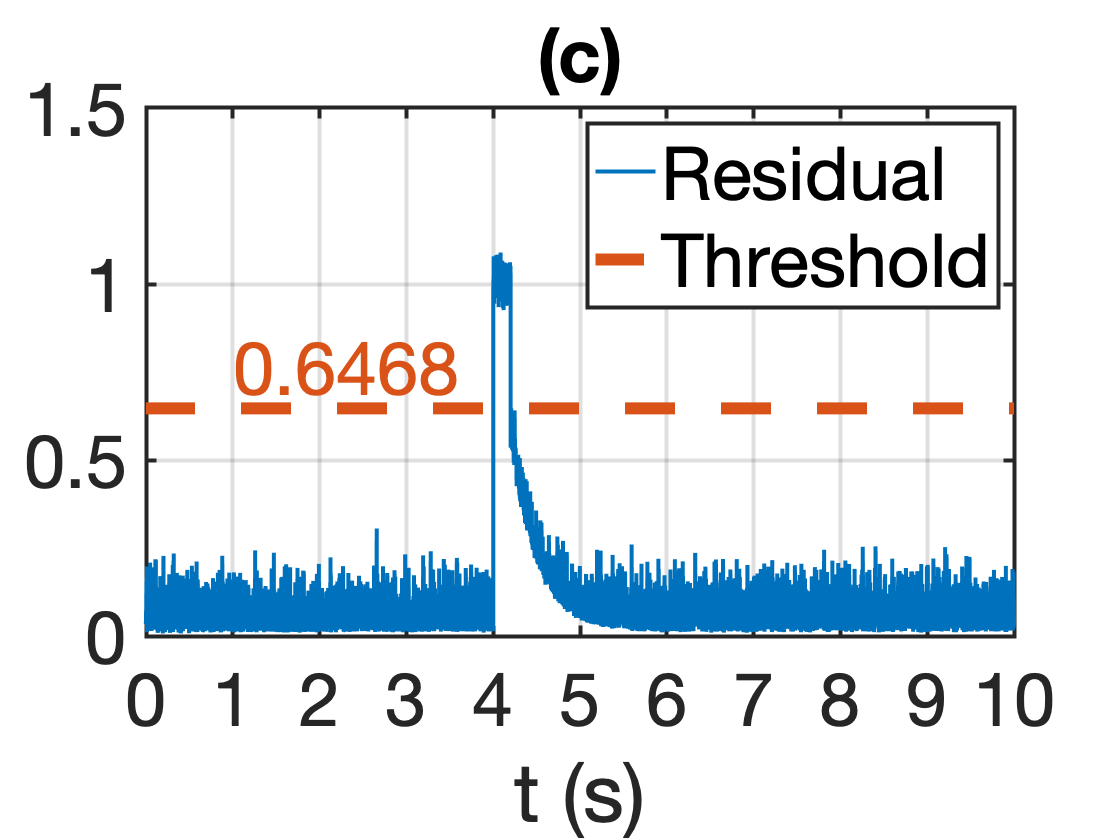}
    \label{diverse-intriago-khan-(c)}
  \end{subfigure}%
  \begin{subfigure}{0.23\textwidth}
    \centering
    \includegraphics[width=1.0\linewidth]{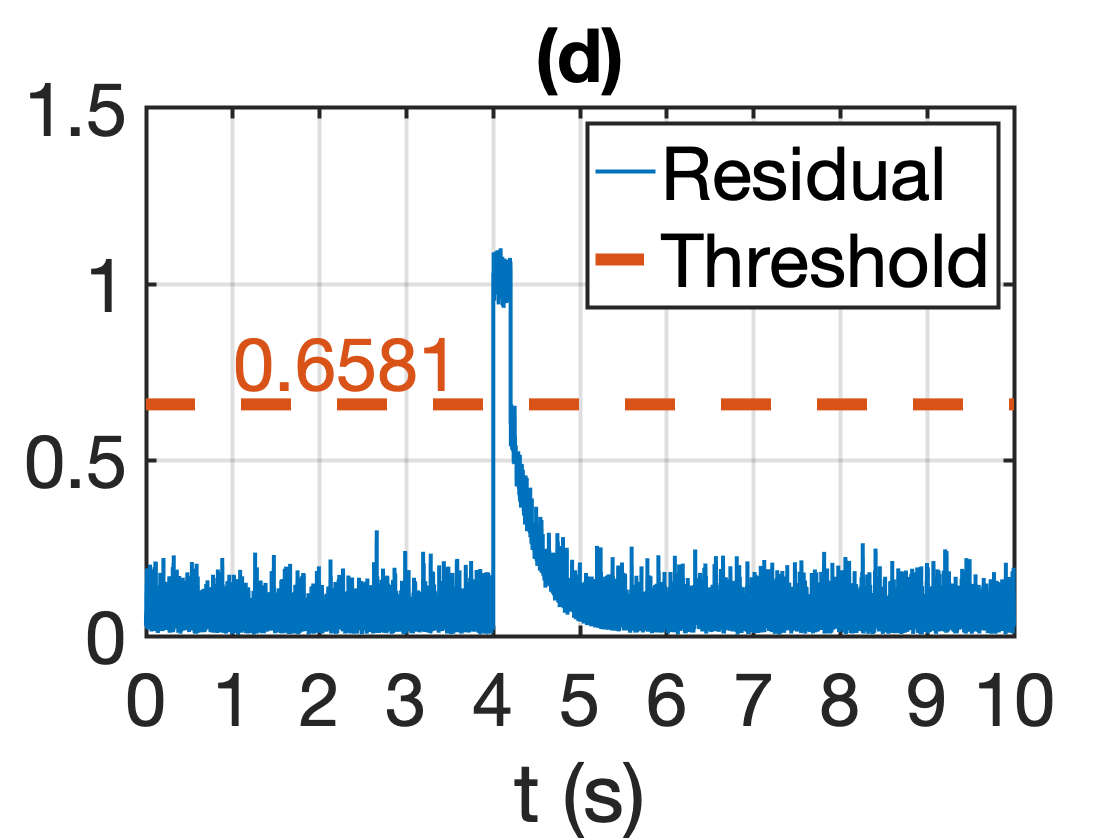}
    \label{diverse-intriago-khan-(d)}
  \end{subfigure}%
  \vspace{-0.3cm}
  \begin{subfigure}{0.23\textwidth}
    \centering
    \includegraphics[width=1.0\linewidth]{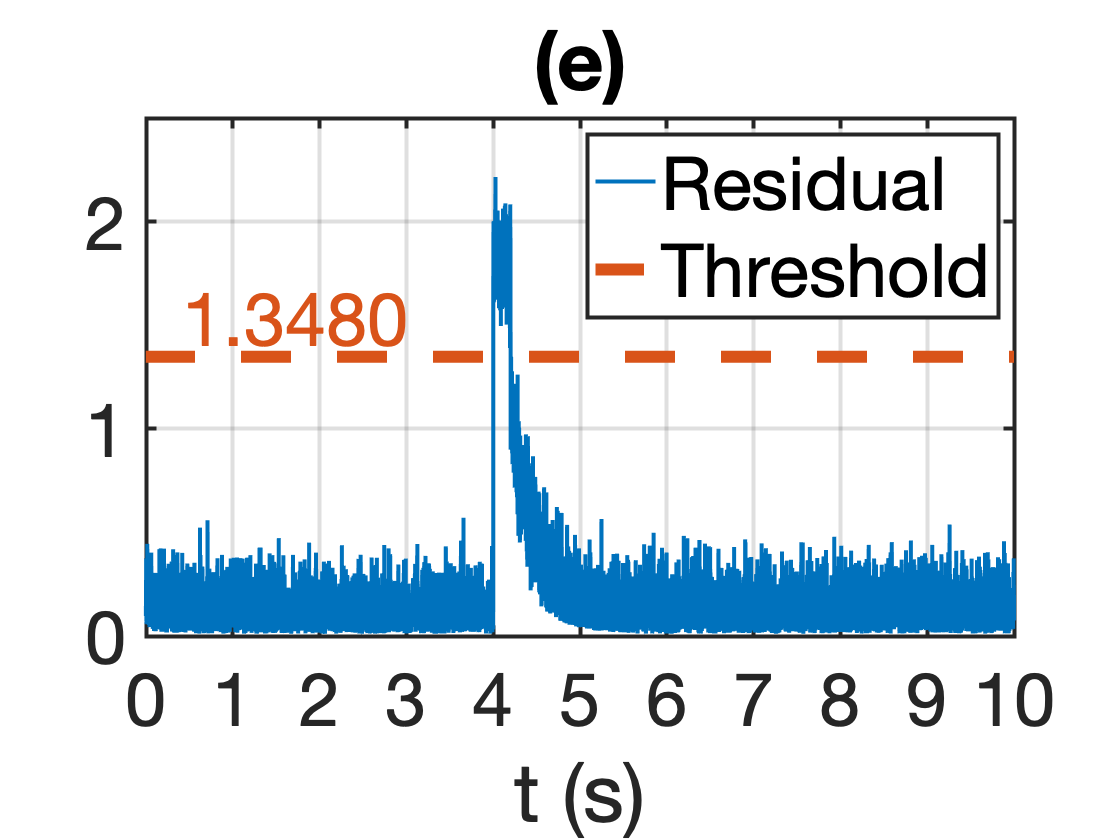}
    \label{diverse-intriago-khan-(e)}
  \end{subfigure}%
  \begin{subfigure}{0.23\textwidth}
    \centering
    \includegraphics[width=1.0\linewidth]{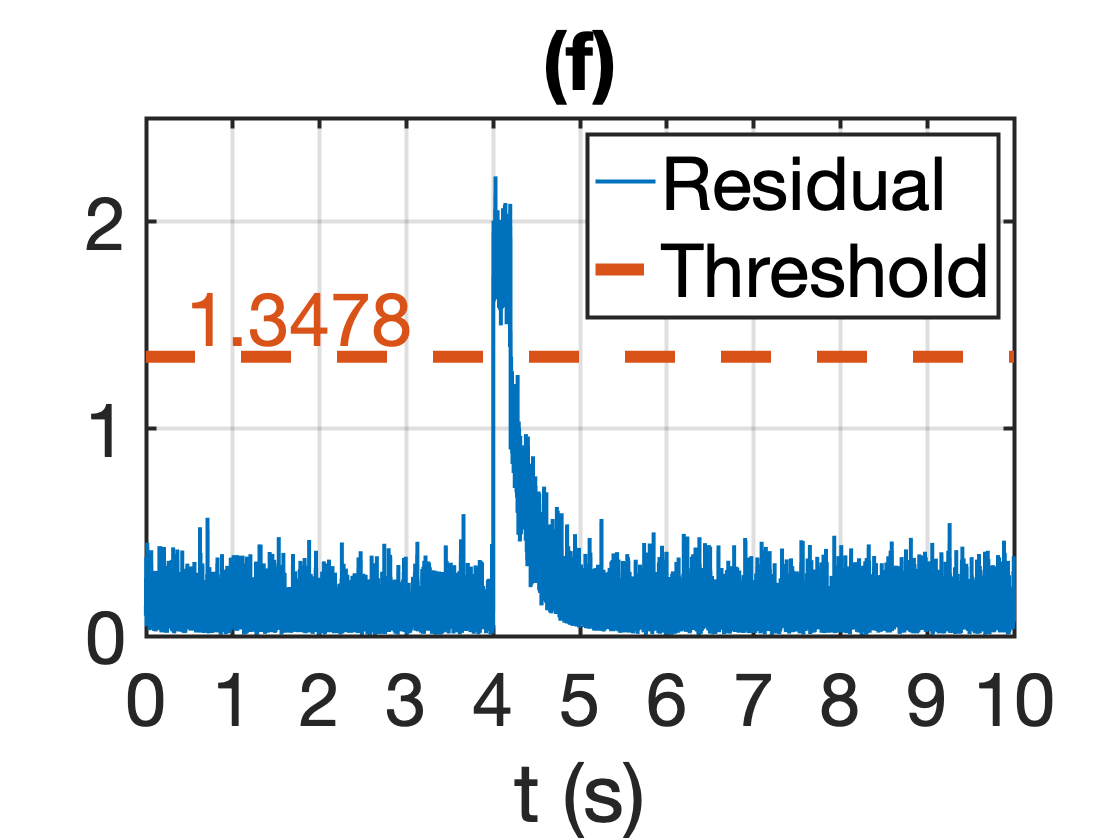}
    \label{diverse-intriago-khan-(f)}
  \end{subfigure}%
  \vspace{-0.3cm}
  \begin{subfigure}{0.23\textwidth}
    \centering
    \includegraphics[width=1.0\linewidth]{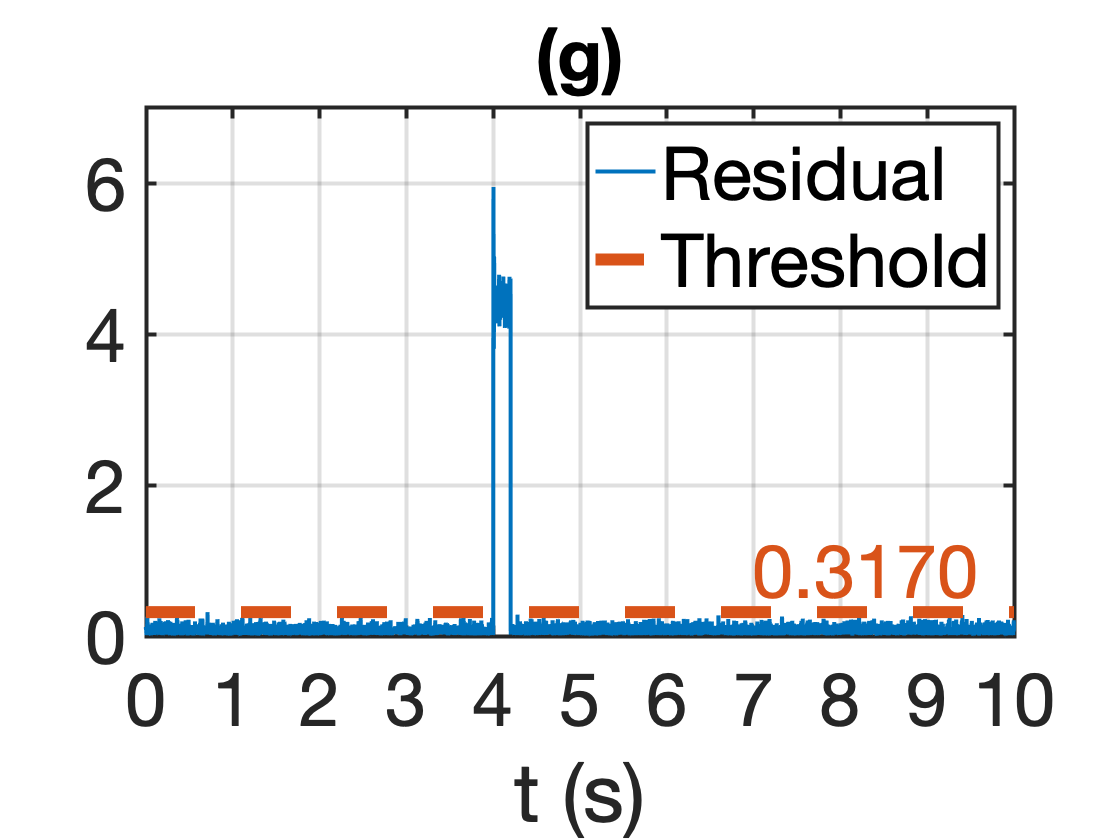}
    \label{diverse-intriago-khan-(g)}
  \end{subfigure}%
  \begin{subfigure}{0.23\textwidth}
    \centering
    \includegraphics[width=1.0\linewidth]{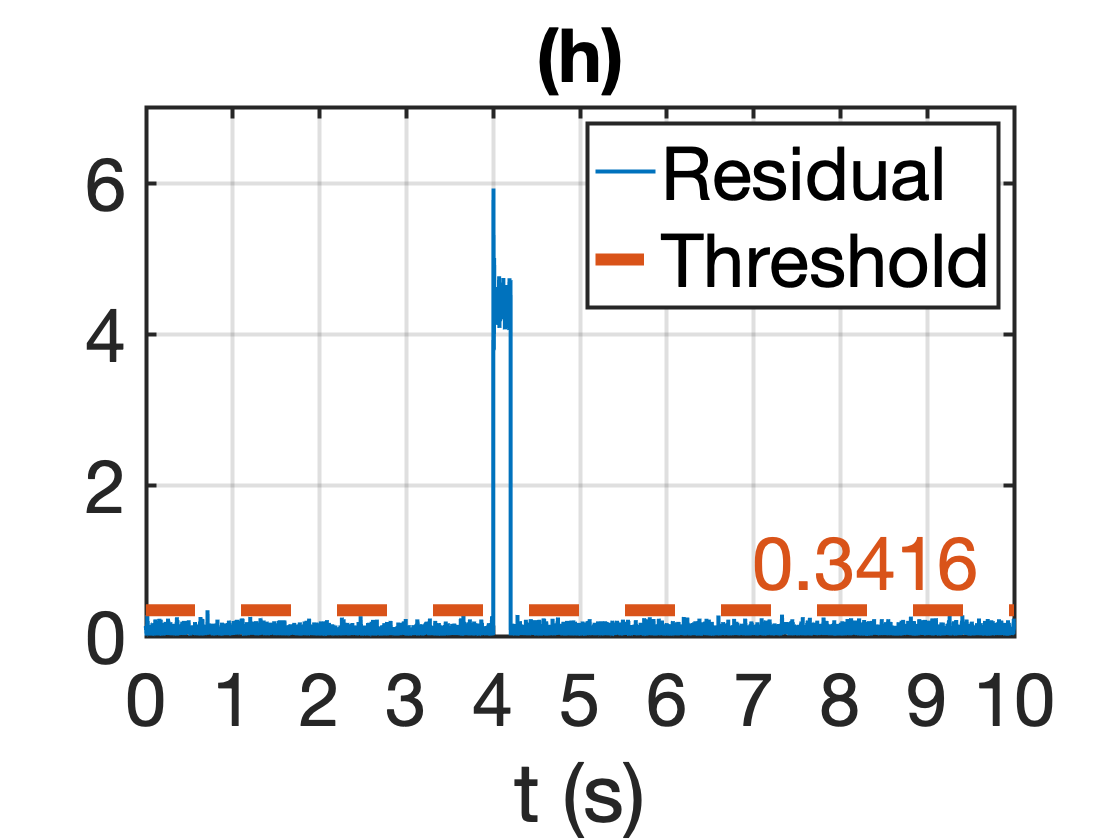}
    \label{diverse-intriago-khan-(h)}
  \end{subfigure}%
  \vspace{-0.4cm}
\caption{Residual norms using an OL-QB observer (left column) and Lipschitz observer (right column) in an AC microgrid with technology mix. (a) and (b) GFM \#1 under a busbar fault; (c) and (d) GFM \#3 under a $\omega_{ni}$ actuator fault; (e) and (f) GFM \#1 under a $V_{ni}$ actuator fault; (g) and (h) GFM \#3 under an inverter bridge fault. The faults last $0.2$ seconds.}
\label{diverse-intriago-khan}
\vspace{-0.4cm}
\end{figure}

\begin{figure}
\centering
  \begin{subfigure}{0.23\textwidth}
    \centering
    \includegraphics[width=1.0\linewidth]{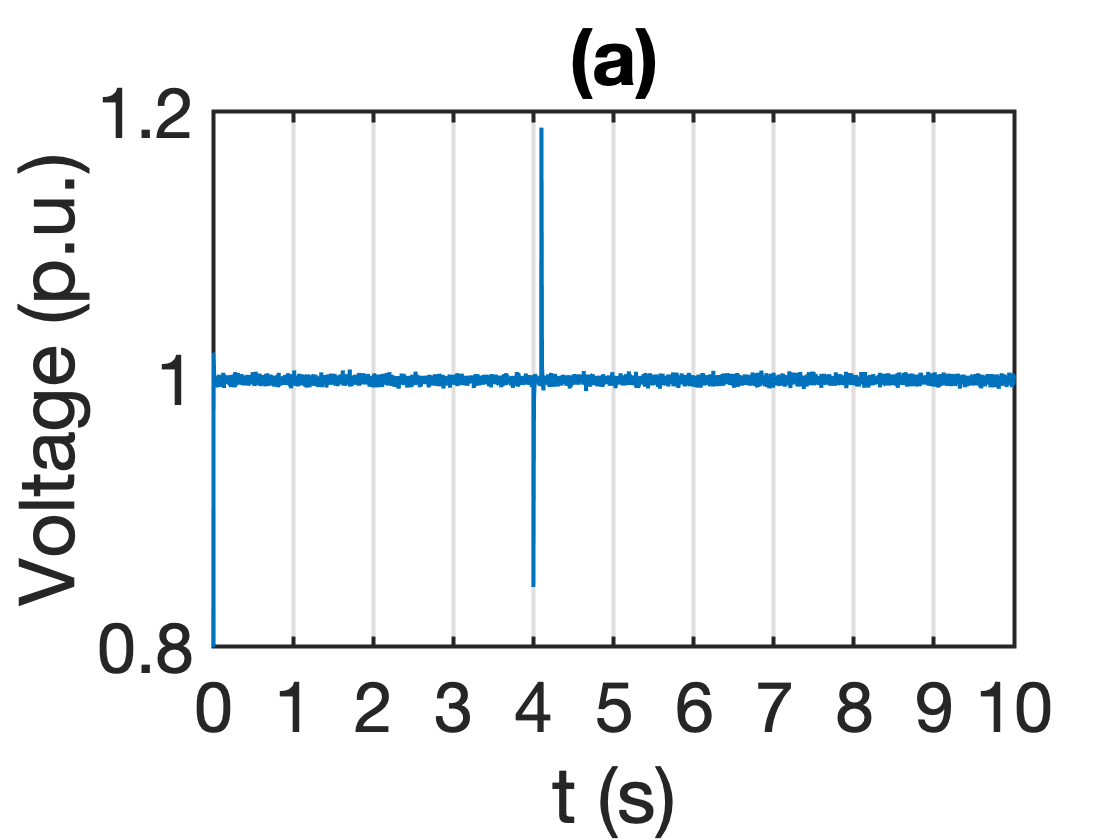}
  \end{subfigure}%
  \begin{subfigure}{0.23\textwidth}
    \centering
    \includegraphics[width=1.0\linewidth]{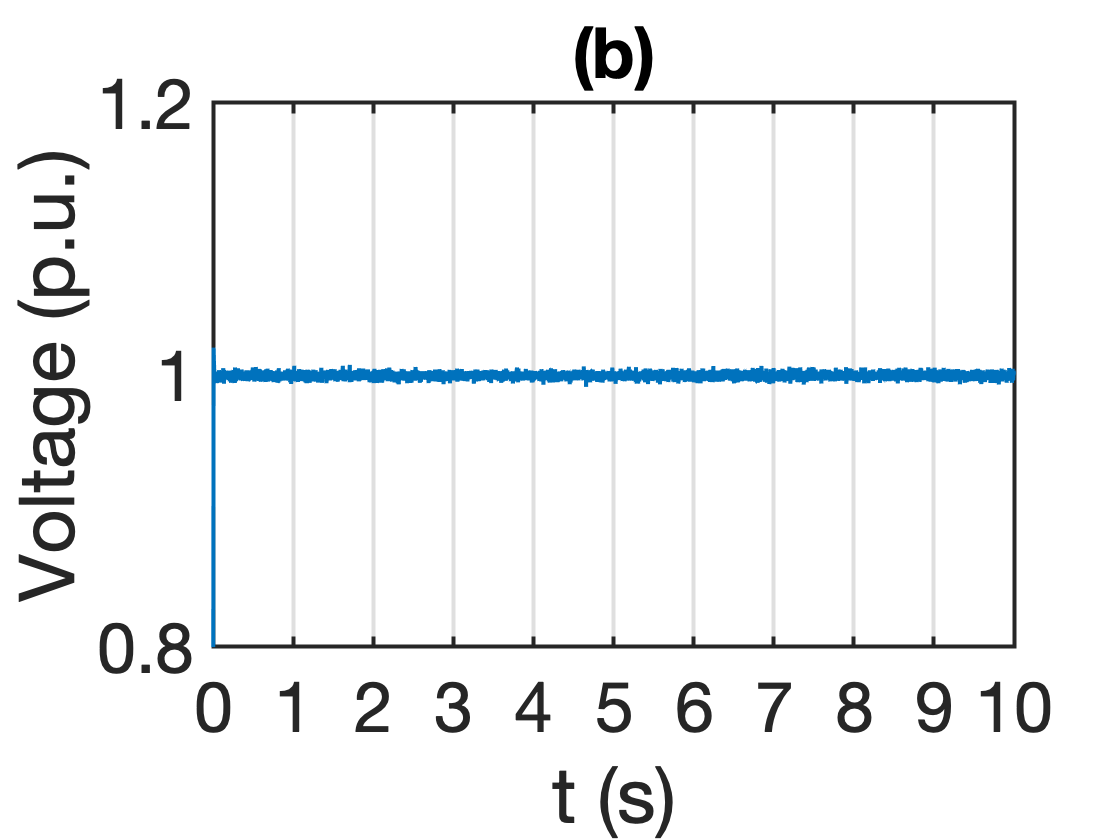}
  \end{subfigure}%
  
  \begin{subfigure}{0.23\textwidth}
    \centering
    \includegraphics[width=1.0\linewidth]{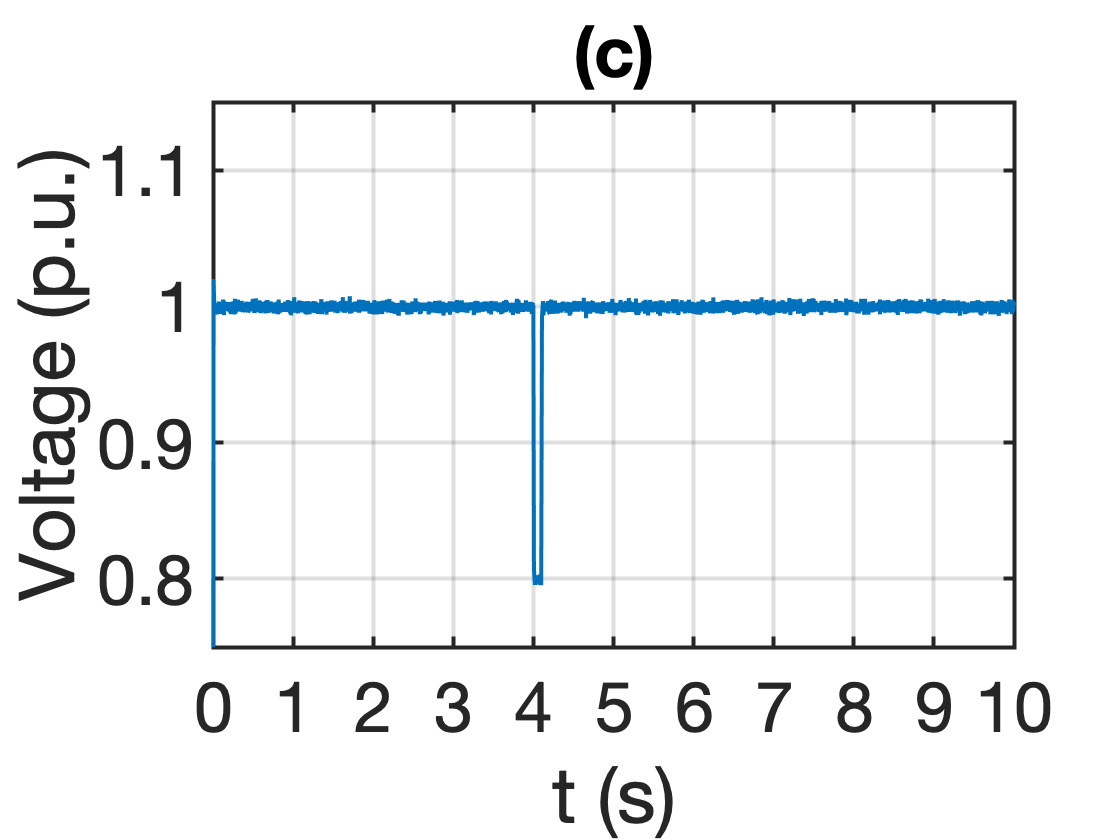}
  \end{subfigure}%
  \begin{subfigure}{0.23\textwidth}
    \centering
    \includegraphics[width=1.0\linewidth]{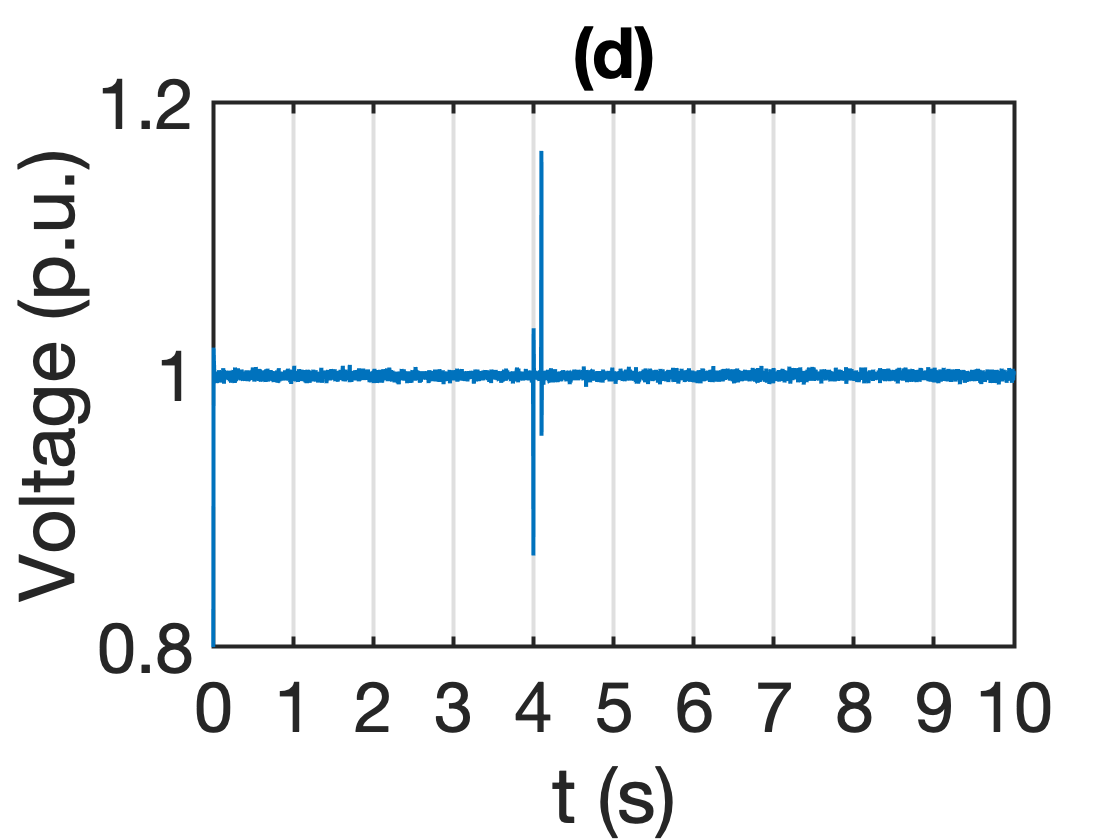}
  \end{subfigure}%
\caption{Voltage magnitude (p.u.) of GFM \#1 in an AC microgrid with technology mix. (a) busbar fault, (b) $\omega_{ni}$ actuator fault, (c) $V_{ni}$ actuator fault, and (d) inverter bridge fault. The faults last $0.2$ seconds.} 
\label{diverse-gfm1-vmag}
\end{figure}

\begin{figure}
\centering
  \begin{subfigure}{0.23\textwidth}
    \centering
    \includegraphics[width=1.0\linewidth]{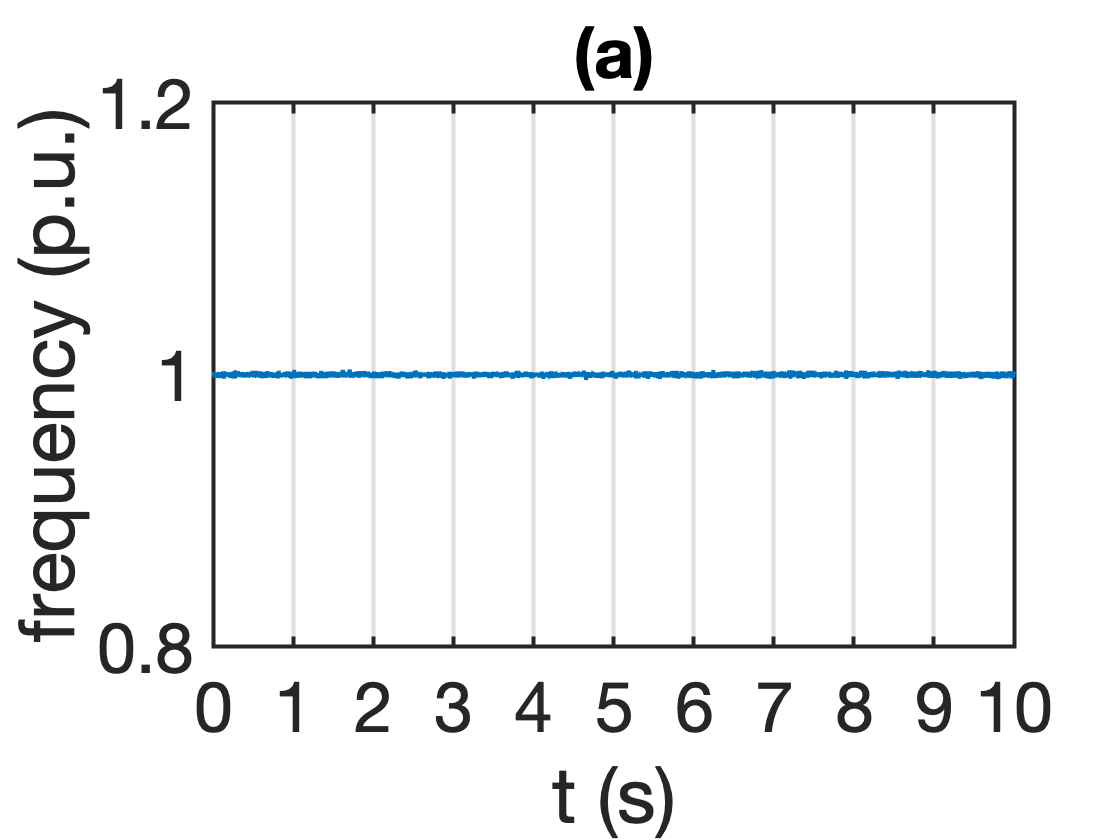}
  \end{subfigure}%
  \begin{subfigure}{0.23\textwidth}
    \centering
    \includegraphics[width=1.0\linewidth]{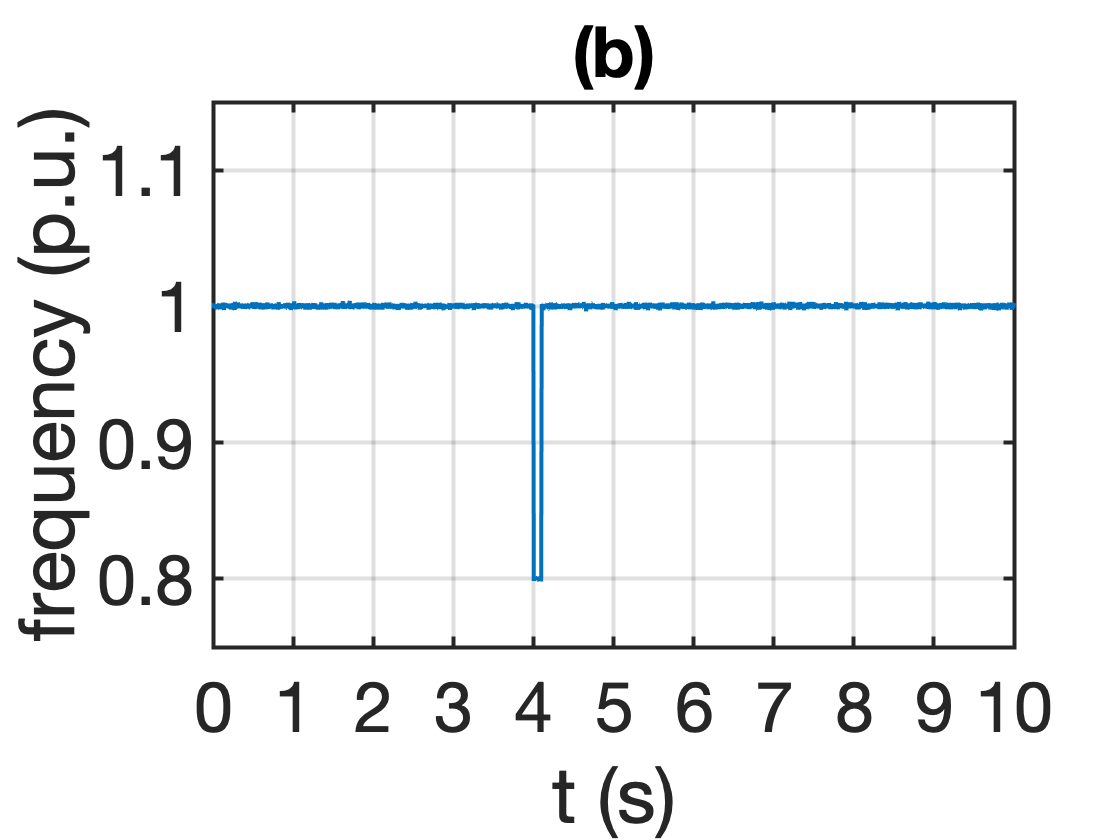}
  \end{subfigure}%
  
  \begin{subfigure}{0.23\textwidth}
    \centering
    \includegraphics[width=1.0\linewidth]{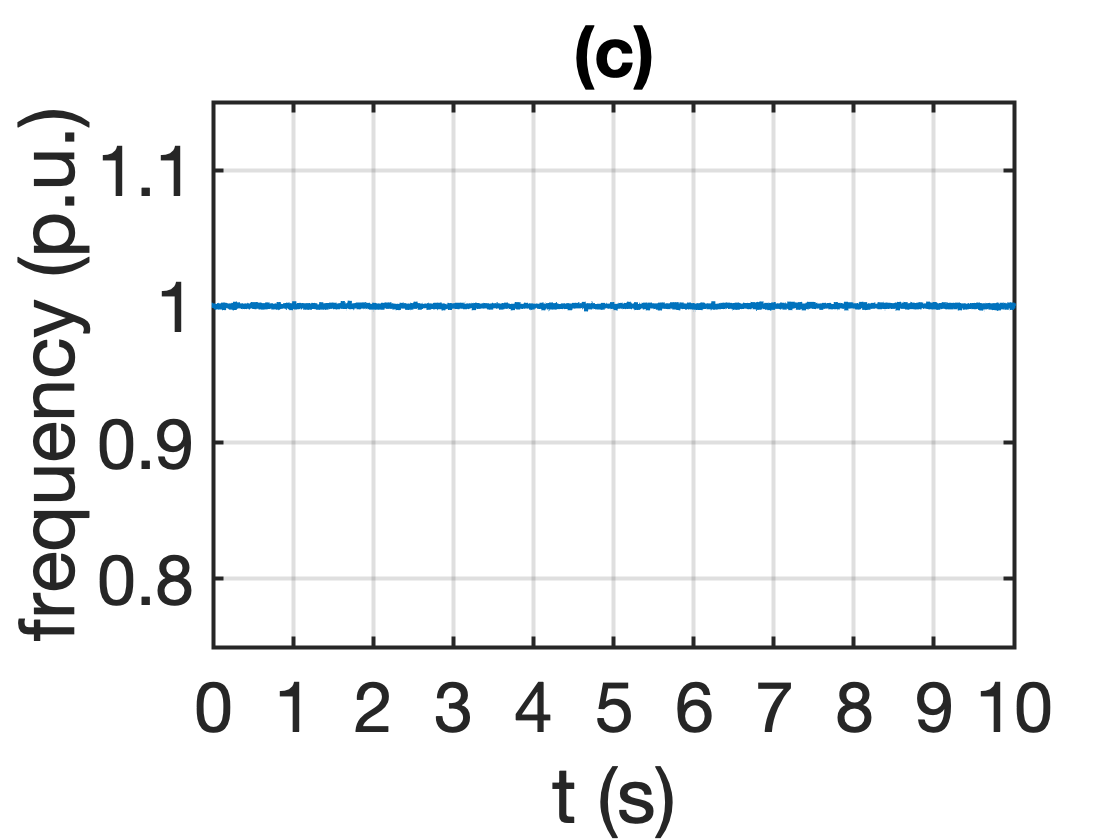}
  \end{subfigure}%
  \begin{subfigure}{0.23\textwidth}
    \centering
    \includegraphics[width=1.0\linewidth]{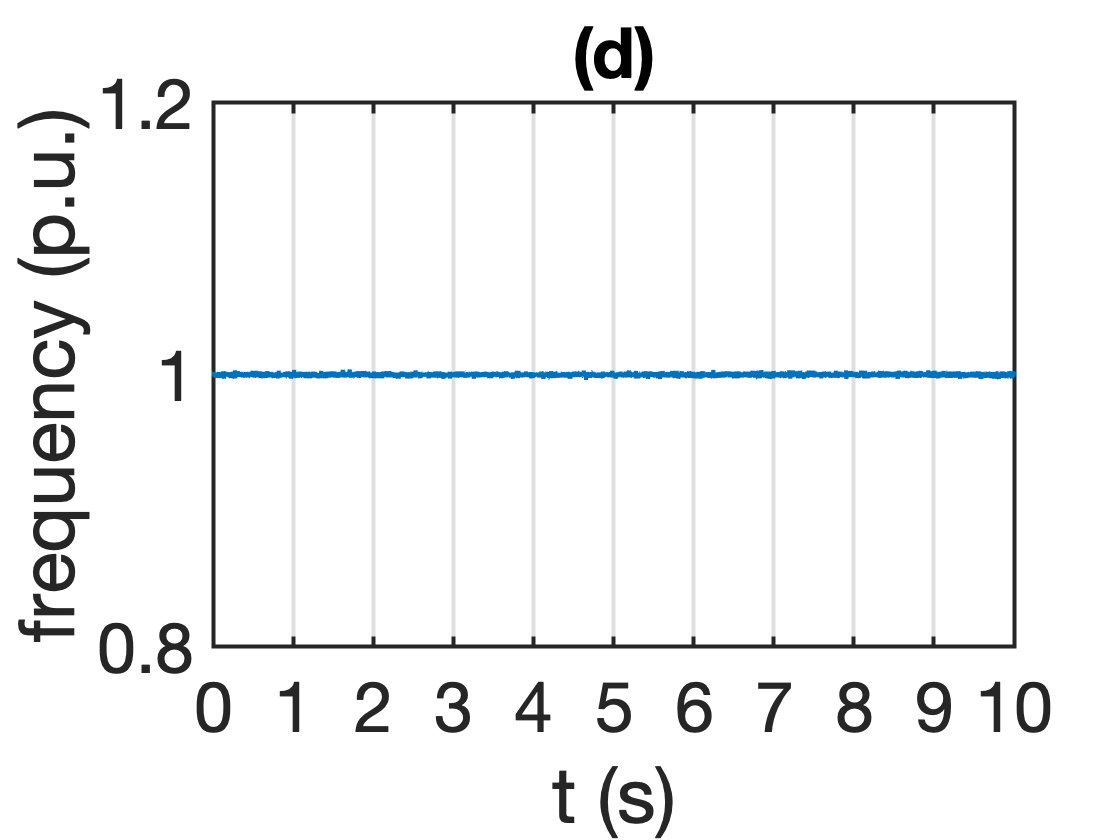}
  \end{subfigure}%
\caption{Frequency (p.u.) of GFM \#1 in an AC microgrid with technology mix. (a) busbar fault, (b) $\omega_{ni}$ actuator fault, (c) $V_{ni}$ actuator fault, and (d) inverter bridge fault. The faults last $0.2$ seconds.} 
\label{diverse-gfm1-freq}
\end{figure}

In this section, we modify the technology mix of the microgrid by replacing GFM \#2 and GFM \#4 with a grid-following inverter (GFL) and a fourth-order synchronous machine. The dynamic equations of the GFL and parameters can be found in \cite{BidramMulti2014}. Similarly, the synchronous machine (SM) parameters and model with turbine-governor and exciter control are found in \cite{Ghahremani2011,neplan,powerworld}. Second, the droop control of the remaining GFMs \#1 and \#3 is transformed into its equivalent virtual synchronous machine control \cite{Sajadi2022}. 

\subsubsection{Busbar fault}
Figure~{\ref{diverse-intriago-khan}a} shows that the residual norm of the OL-QB observer has increased its peak magnitude during the presence of the busbar fault while preserving the threshold's magnitude compared with the norm's response using droop control as presented in Figure~{\ref{one-sided-lipschitz-3ph}a}. The residual norm of the Lipschitz observer exhibits unresponsive behavior during the fault occurrence. The OL-QB observer exhibits a fault detection time around $t_{o} = 9$ ms, whereas the clearing time is close to $t_{c} = 50$ ms.


\subsubsection{Actuator faults}
Figure~{\ref{diverse-intriago-khan}} shows that the shapes of the residual norm under the $\omega_{ni}$ and $V_{ni}$ actuator faults present a slow decay after the fault clearance. Such a decay demands an increase in the threshold for the OL-QB and Lipschitz observers. According to our experiments, the slow decay is a consequence of the dynamics from the frequency and voltage magnitude secondary controllers. The detection time for both type of observers is $t_{o} = 1$ ms, whereas the clearing times are less than $t_{o} = 1.3$ ms.

\subsubsection{Inverter bridge fault}
The response of the residual norm under the inverter bridge fault for GFM \#3 is shown in Figures~{\ref{diverse-intriago-khan}g} and {\ref{diverse-intriago-khan}h}. Both responses show higher peak values during the presence of the fault compared to the peak values presented in Figure~{\ref{bridge}}. The threshold magnitudes remain similar to the magnitudes of the droop-controlled GFM \#4 shown in Figures~{\ref{bridge}c} and {\ref{bridge}d}. Both observers present fault detection and clearing times of less than 1 ms.

Figures \ref{diverse-gfm1-vmag} and \ref{diverse-gfm1-freq} show the response of the voltage magnitude and frequency in per unit at the PCC of GFM \#1 when the four internal faults happen. A remarkable observation is that faults are likely to be decoupled between voltage magnitude and frequency. For example, the busbar fault is clearly observable in the voltage magnitude, not the frequency. The reason for this phenomenon lies in the characteristic swing equation where the equivalent inertia is inversely proportional to the droop coefficient $m_{Pi}$, which is in the order of $10^{-5}$. Consequently, any change in the unfiltered active power due to the busbar fault does not largely influence the dynamics of frequency deviations $\Delta\omega$. A similar analysis can be derived for the other internal faults and the droop-controlled GFMs.

\begin{table}[t]
\caption{Minimum and maximum fault detection ($t_o$) and clearance ($t_c$) response times of the proposed OL-QB observers. The values are obtained by comparing the response times of the observers among all the experiments.}
\begin{center}
\begin{tabular}{ccccccccc} 
\hline
    & \multicolumn{8}{c}{Fault} \\
    & \multicolumn{2}{c}{Busbar} & \multicolumn{2}{c}{$\omega_{ni}$} & \multicolumn{2}{c}{$V_{ni}$} & \multicolumn{2}{c}{Bridge} \\
    & $t_o$ & $t_c$ & $t_o$ & $t_c$ & $t_o$ & $t_c$ & $t_o$ & $t_c$ \\
\hline
\hline
min & 7.5   & 24.9  & 1.0   & 1.1   & 1.0   & 1.1   & 1.0   & 1.0   \\ 
max & 49.7  & 52.2  & 1.0   & 1.3   & 1.0   & 1.2   & 1.0   & 1.0   \\ 
\hline
\end{tabular}
\end{center}
\label{response-times-proposed}
\vspace{-0.5cm}
\end{table}

Table \ref{processing-time} shows a comparison of the simulation times between the OL-QB and Lipschitz observers. We can not report the simulation time of the Lipschitz design under a busbar fault because its corresponding optimization program (cf. Lemma \ref{lipschitz-proposition}) that computes the observer gain matrix is reported as infeasible by the optimization solvers (cf. Remark~\ref{remark3}). Table {\ref{response-times-proposed}} presents the minimum and maximum fault detection and clearing times of our approach among all the simulations.

\begin{remark}\label{remark3}

We identify a drawback of the state-of-the-art Lipschitz observer design. For this particular design, notice that the squared Lipschitz constant in the block matrices $\Omega_1$ and $\Omega_2$ forces the positive definiteness of the LMIs in {\eqref{L1-lipschitz}} and {\eqref{L2-lipschitz}}. Hence, the feasible set for the matrix gain $L$ becomes more restricted than the counterpart of the one-sided Lipschitz observer. The situation is further aggravated by the magnitude of the Lipschitz constant. The feasibility issue becomes noticeable for the busbar faults. The four solvers report the optimization problem described by {Lemma} {\ref{lipschitz-proposition}} as infeasible. On the other hand, the OL and QB constants are not squared in their corresponding constraints and are at least half of the Lipschitz constant (cf. Table {\ref{constants}}).

\end{remark}

\subsection{Simultaneous occurrence of faults}\label{simultaneous}
\begin{figure}
\centering
  \begin{subfigure}{0.23\textwidth}
    \centering
    \includegraphics[width=1.0\linewidth]{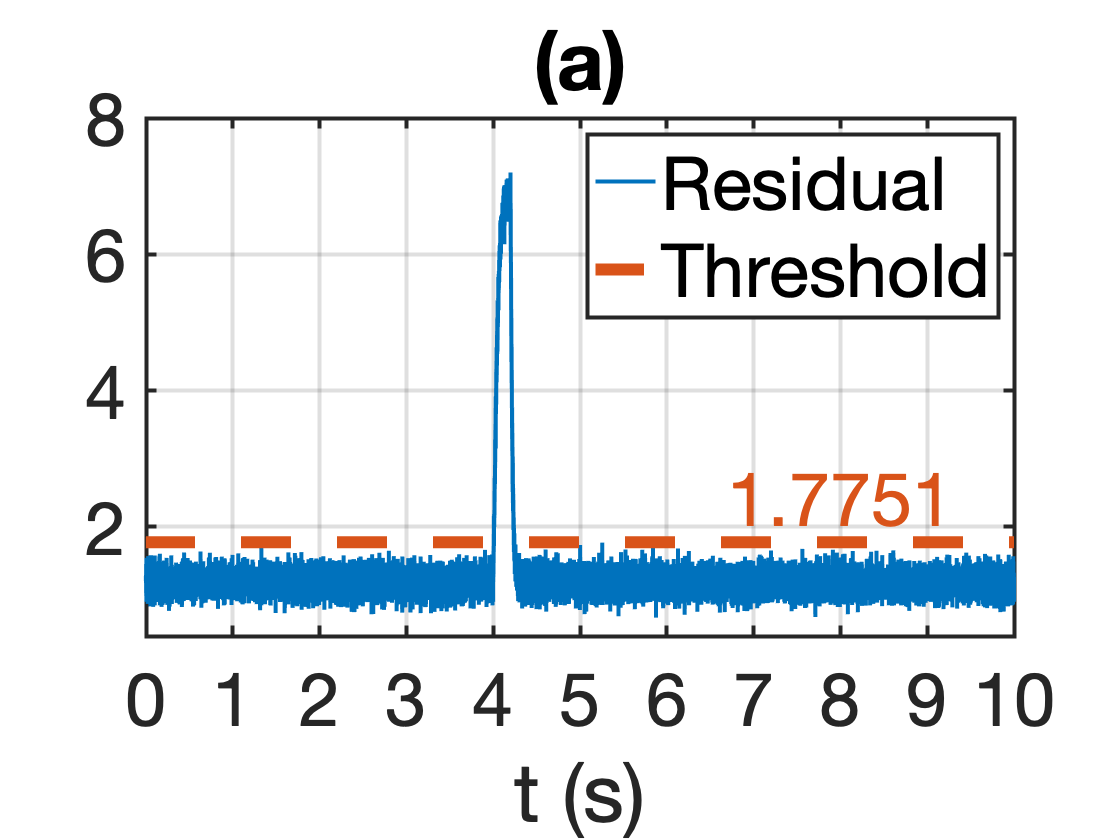}
  \end{subfigure}%
  \begin{subfigure}{0.23\textwidth}
    \centering
    \includegraphics[width=1.0\linewidth]{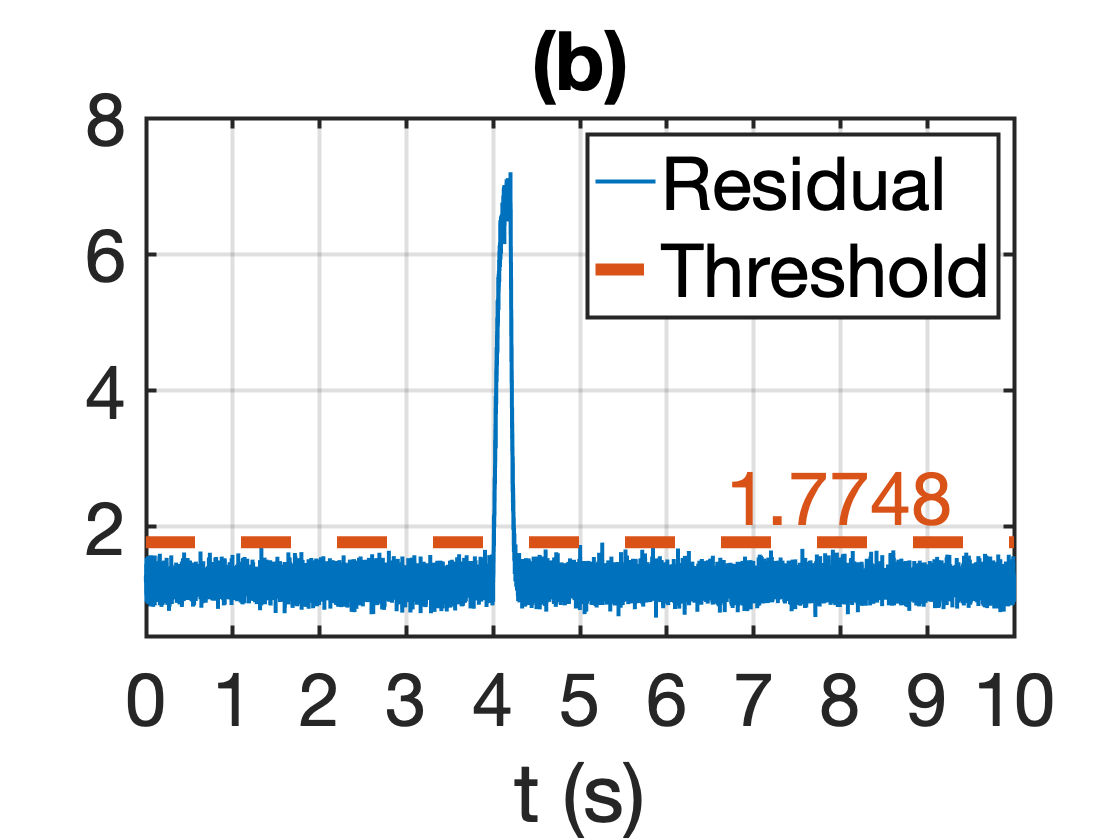}
  \end{subfigure}%
  \vspace{0.2cm}
  \begin{subfigure}{0.23\textwidth}
    \centering
    \includegraphics[width=1.0\linewidth]{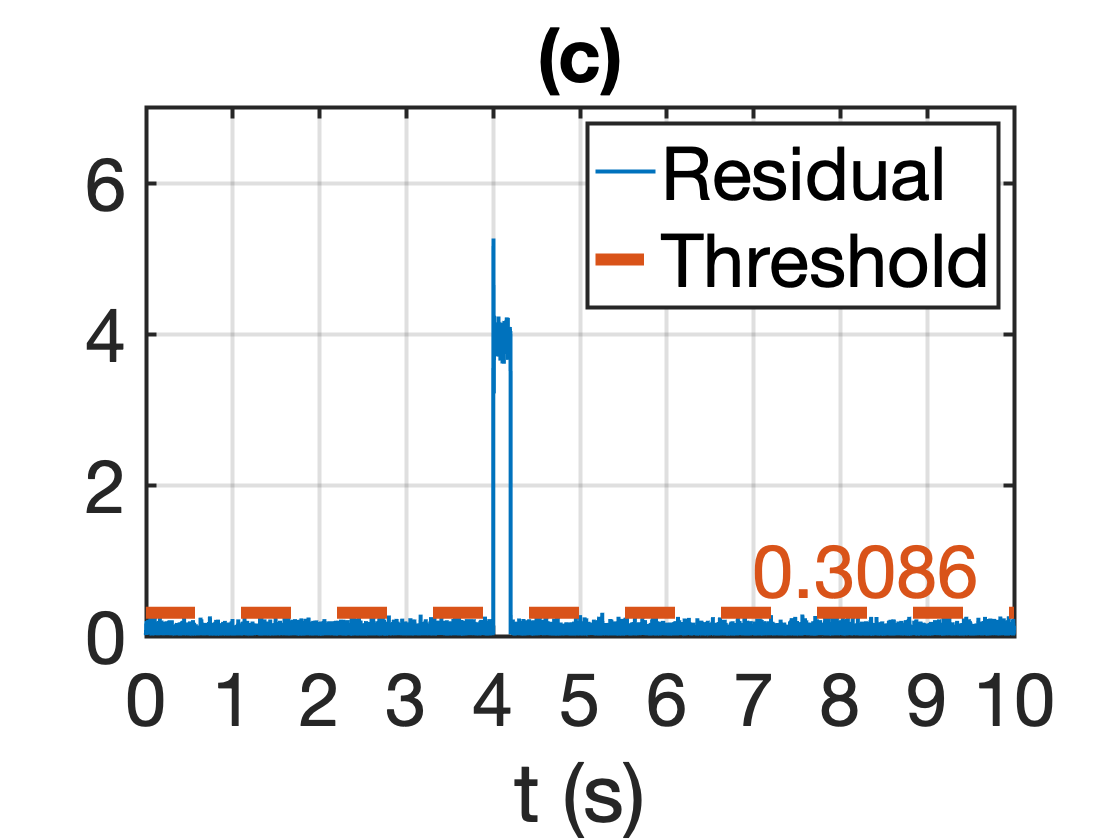}
    \label{simultaneous-vsm-bridge-observer-d}
  \end{subfigure}%
  \begin{subfigure}{0.23\textwidth}
    \centering
    \includegraphics[width=1.0\linewidth]{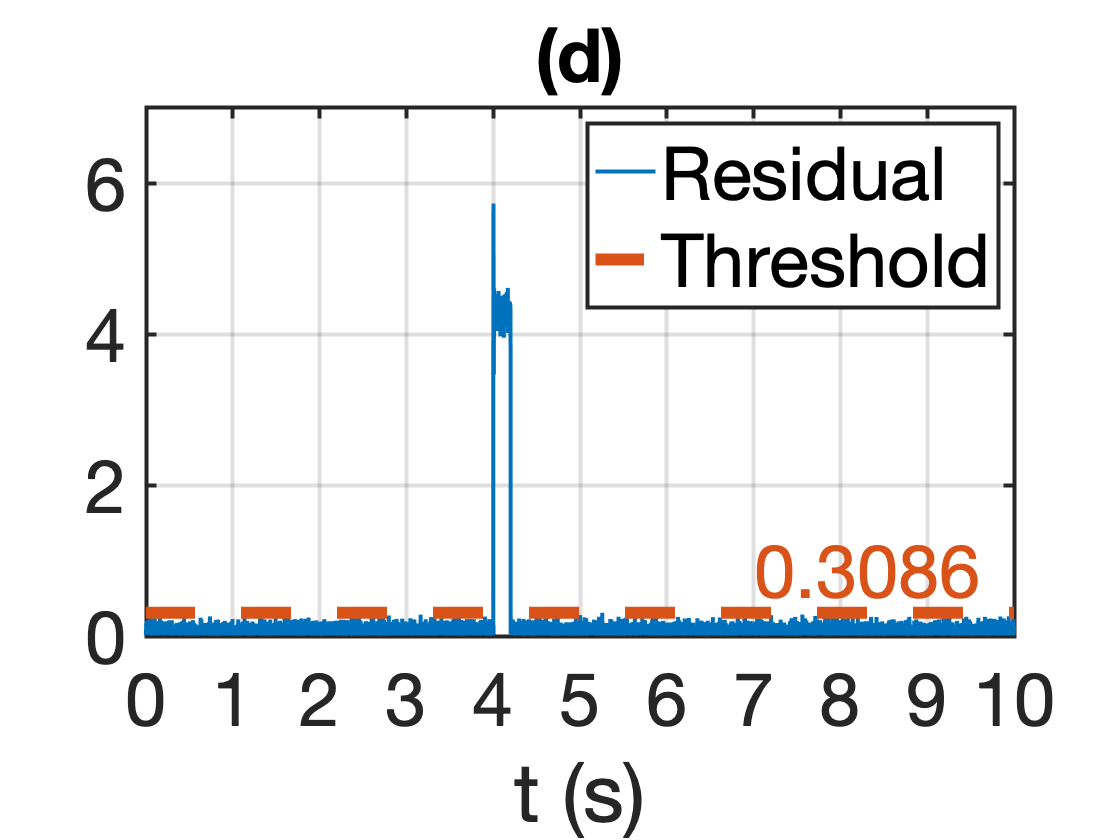}
    \label{simultaneous-vsm-bridge-observer-f}
  \end{subfigure}%
  \vspace{-0.4cm}
\caption{Residual norms under the simultaneous occurrence of faults in GFM $\#1$ using the proposed OL-QB observer with droop control (top row) and Lipschitz observer with virtual synchronous machine (bottom row). (a) busbar and actuator $\omega_{ni}$ fault; (b) busbar and both actuator $\omega_{ni}$ and $V_{ni}$ faults; (c) inverter bridge and actuator $V_{ni}$ fault; (d) inverter bridge, busbar and actuator $V_{ni}$ faults. The simultaneous faults occur at $t=4$ and are cleared after $0.2$ seconds.} 
\label{simultaneous-faults}
\vspace{-0.1cm}
\end{figure}

Figure~{\ref{simultaneous-faults}} presents a set of experiments under the simultaneous occurrence of faults for GFM \#1. Figure~{\ref{simultaneous-faults}a} shows the response of the residual norm under the simultaneous presence of a busbar and {$V_{ni}$} faults using our proposed observer designed for a busbar fault, whereas an additional {$\omega_{ni}$} fault is added to the experiment in Figure~{\ref{simultaneous-faults}b}. The busbar fault is the major contributor to the sensitivity of the residual signal. It can be seen that the {$\omega_{ni}$} fault plays an almost negligible role in the sudden change of the residual norm. Figure~{\ref{simultaneous-faults}c} shows the response of the residual norm under the presence of an inverter bridge and actuator {$V_{ni}$} faults using a Lipschitz observer designed for an inverter bridge fault; an additional busbar fault is added to the experiment in Figure~{\ref{simultaneous-faults}d}. The residual signal exhibits a slight increase due to the presence of the busbar fault. Notice that regardless of the simultaneous presence of the faults in the previous four experiments, we observe that the busbar and inverter bridge faults can be identified precisely.

\begin{remark}
Fault diagnosis is a two-stage process consisting of fault detection and fault location stages. The fault detection stage sets the basis for a reliable fault location technique. A fault location strategy can be easily derived once a fault is detected. It is worth noting that a single fault may perturb the residual norms of different observers. In this regard, we propose employing a bank of observers to identify triggering patterns among them, thus mitigating the occurrence of false alarms. Another approach is to incorporate the matrix expressions $E_f$ and $F_f$ corresponding to other faults into the disturbance matrix expressions $E_w$ and $F_w$ of the specific fault being analyzed. Alternatively, a combination of these two solutions, as outlined in \mbox{\cite{Wang2021}}, can also be considered. However, the comprehensive exploration of these directions is beyond the scope of this paper and will be pursued in future research endeavors.
\end{remark}

\vspace{-0.5cm}

\section{Discussion}\label{discussion}
\subsection{Complexity and computational burden}
We divide our proposed scheme into four important stages to define its complexity and computation burden. The first stage computes the OL and QB  constants in less than 20 minutes according to the algorithms in \mbox{\cite{Qi2018}} with a complexity of $\mathcal{O}(n^3\:n_{\mathcal{D}})$ and $\mathcal{O}(n^3\:n_{\mathcal{D}}^2)$ respectively, where $n$ is the number of states and $n_{\mathcal{D}}$ is the number of samples taken from the operating region $\mathcal{D}$. The second stage solves the semidefinite program posed in Theorem 1 in less than twenty seconds using primal-dual interior points methods, which have a worst-case complexity estimate of $\mathcal{O}(p^{2.75}K^{1.5})$ where $p$ is the number of variables (states and measurements), and $K$ is the number of constraints \mbox{\cite{Boyd1994}}. The threshold computation performed in the third stage has a lesser processing burden than previous stages with a complexity of $\mathcal{O}(t)$ proportional to the finite-length period $t$ discussed in Section {\ref{threshold-residual-computation}}. Lastly, the fourth stage computes the residual norm online with a complexity of $\mathcal{O}(m)$ where $m$ is the number of the system's measurements.


\vspace{-0.4cm}

\subsection{Comparison between similar fault detection schemes}
The works \mbox{\cite{Wang2021}}, \mbox{\cite{Shoaib2022}}, and \mbox{\cite{Mehmood2023}} propose methods similar to our proposed approach because they develop a residual-, observer-, and model-based fault detection strategy using the $\mathcal{H}_{-}/\mathcal{H}_{\infty}$ optimization framework. However, the main limitation of \mbox{\cite{Wang2021}} is that the observers' design targets linear dynamic systems in DC microgrids. The essential drawback of \mbox{\cite{Shoaib2022}} is that the authors need to consider the limitations of the Lipschitz condition while designing the observers for fault detection. In \mbox{\cite{Mehmood2023}}, the authors limit their approach to sensor faults and do not study grid-forming inverters. Our proposed fault detection method overcomes these limitations because the observers are designed for the nonlinear dynamics of the GFMs in AC microgrids, considering four different types of faults based on the OL-QB conditions that improve upon the limitation of the state-of-the-art Lipschitz design.

\vspace{-0.3cm}

\subsection{Applicability to inverter-dominant large-scale power systems}
The proposed methodology needs further analysis before its immediate application to inverter-dominant large-scale power systems due to the significant stability challenges. The dynamics of numerous converters with different controls may operate on a similar time scale as the dynamics of the lines, resulting in resonance phenomena and, ultimately, instability \mbox{\cite{Dorfler2023}}. Large-scale power systems may demand longer interconnection lines, which is detrimental to the stable penetration of inverter-based resources \mbox{\cite{Ding2021}}. Potential changes in the grid-forming inverters at the parametric and control levels are required to accommodate the aforementioned stability challenges. We recommend considering a sensitivity analysis of the inverter's parameters to study the feasibility of the LMIs in Theorem {\ref{one-sided-lipschitz-theorem}}. Also, we consider it essential to study the impact of new controller designs over the OL-QB constants.

\subsection{Capacitive lines}
We investigate how our fault detection algorithm should vary if lines are capacitive. Our findings reveal that lines with more than 0.1 $mF$ capacitances impact the proposed algorithm, especially the observers for busbar faults. We notice that the observers for the other types of faults are less impacted by the capacitive lines. The reason is that the capacitive lines interact with the inverters' voltage controllers by exchanging reactive power. The interaction is observed as an oscillation in the observers' residual signal. We suggest modifying our proposed approach by including a low-pass filter applied to the residual norm and/or changing the threshold value to attenuate the impact of oscillations.



\section{Conclusion}\label{conclusion}
This work proposes a nonlinear observer design based on the OL and QB conditions for detecting internal faults in the emergent grid-forming inverters technology. We derive the matrix expressions of faults and disturbances affecting GFMs and consider them in the observer design process. The internal faults considered in this work are busbar, actuator, and inverter bridge faults. The nonlinear model of the GFM is expressed as a one-sided Lipschitz formulation. We pose a set of LMI constraints and $\mathcal{H}_{-}/\mathcal{H}_{\infty}$ optimization to design an observer that achieves sensitivity to faults and robustness against disturbances. The proposed observer design is compared with the state-of-the-art design based on the Lipschitz condition. The association between the Lipschitz and the OL-QB observer is studied theoretically and experimentally, showing that the latter allows for a less restrictive observer design with less computational time. We demonstrate that our proposed approach performs better than the state-of-the-art designs regarding trustworthy detection and computational time. Moreover, leveraging the available signals of typical grid-forming inverters, our approach does not require additional sensors yielding a cost-effective solution. Our proposed approach is evaluated with an islanded droop-controlled microgrid. The numerical tests corroborate our proposed approach's effectiveness and feasibility of real-world implementation.
\vspace{-0.5cm}

\bibliographystyle{IEEEtran}
\bibliography{references}
\end{document}